\documentclass[hidelinks,11pt]{article}

\usepackage{amsmath}
\usepackage{amssymb}
\usepackage{amsfonts,amsthm}
\usepackage{thmtools,thm-restate}
\usepackage{tabularx,environ}
\usepackage{multirow}
\usepackage{hyperref}
\usepackage{graphicx,float}
\usepackage{enumitem,linegoal}
\usepackage{caption}
\usepackage{subcaption}
\usepackage{enumerate}
\usepackage{complexity}
\usepackage{cleveref}

\usepackage{boxedminipage,tikz}

\setlist[itemize]{noitemsep}
\setlist[enumerate]{noitemsep}

\usepackage{authblk}

\newtheorem{theorem}{Theorem}
\newtheorem{lemma}[theorem]{Lemma}

\newtheorem{claim}{Claim}

\newtheorem{proposition}[theorem]{Proposition}
\newtheorem{observation}{Observation}

\theoremstyle{definition}

\newtheorem{remark}{Remark}

\newcommand{\Oh}{\mathcal{O}}

\newenvironment{claimproof}{\begin{proof}\renewcommand{\qedsymbol}{\claimqed}}{\end{proof}\renewcommand{\qedsymbol}{\plainqed}}
\let\plainqed\qedsymbol

\usepackage{tikz}
\usetikzlibrary{arrows}
\usetikzlibrary{arrows.meta}
\usetikzlibrary{decorations.pathmorphing}
\usetikzlibrary{decorations.markings}
\usetikzlibrary{calc}
\tikzset{
  circ/.style = {circle,draw,fill,inner sep=1.3pt},
  cir/.style = {circle,draw,fill,inner sep=1pt},
  pt/.style = {circle,draw,fill,inner sep=.5pt},
  circr/.style = {circle,draw=red,fill=red,inner sep=1.3pt},
  scirc/.style = {circle,draw,fill,inner sep=.5pt},
  invisible/.style = {draw=none,inner sep=0pt,font=\tiny},
  nonedge/.style={decorate,decoration={snake,amplitude=.3mm,segment length=1mm},draw}
}

\usepackage[a4paper,bindingoffset=0.2in,
            top=1.1in, bottom=1.1in, left=1.0in, right=1.0in,
            footskip=.5in]{geometry}

\newcommand\yes{\textsc{Yes}}
\newcommand\no{\textsc{No}}

\newcommand\tdom{\textsc{Total Dominating Set}}

\newcommand\contracd{\textsc{1-Edge Contraction($\gamma$)}}

\newcommand\contractd{\textsc{1-Edge Contraction($\gamma_{t}$)}}

\newcommand{\set}[1]{\ensuremath{ \left\lbrace  #1 \right\rbrace }}

\DeclareMathOperator*{\argmax}{arg\,max}

\title{Blocking total dominating sets via edge contractions}

\date{}
\author{E. Galby, F. Mann, B. Ries}

\affil{University of Fribourg, Department of Informatics,\\
	Fribourg, Switzerland}
\begin{document}
\maketitle

\begin{abstract}
In this paper, we study the problem of deciding whether the total domination number of a given graph $G$ can be reduced using exactly one edge contraction (called \contractd\ ). We focus on several graph classes and determine the computational complexity of this problem. By putting together these results, we manage to obtain a complete dichotomy for $H$-free graphs (see Theorem 3).
\end{abstract}


\section{Introduction}

In this paper, we consider the problem of reducing the total domination number of a graph by contracting a single edge. More precisely, given a graph $G=(V,E)$, we want to know whether there exists an edge $e\in E$ such that the total domination number of the graph $G'$, obtained from $G$ by contracting the edge $e$, is strictly less than the total domination number of $G$. This problem fits into the general framework of so-called \textit{blocker problems} which have been studied intensively in the literature (see for instance  \cite{BTT11,bazgan2013critical,RBPDCZ10,Bentz,CWP11,DPPR15,diner2018contraction,PAL-GAL-RIE,PBP,PPR16,paulusma2017blocking,paulusma2018critical}). In this framework, we ask for a specific graph parameter $\pi$ to decrease: given a graph $G$, a set $\mathcal{O}$ of one or more graph operations and an integer~$k\geq 1$, the question is whether $G$ can be transformed into a graph $G'$ by using at most $k$ operations from $\mathcal{O}$ such that $\pi(G')\leq \pi(G)-d$ for some {\it threshold} $d\geq 0$. Such problems are called blocker problems as the set of vertices or edges involved can be viewed as ``blocking'' the parameter $\pi$. Identifying such sets may provide important information about the structure of the graph~$G$. Blocker problems can be seen as a kind of graph modification problems. Indeed, in such problems we are usually interested in modifying a given graph $G$, via a small number of operations, into some other graph $G'$ that has a certain desired property which often describes a certain graph class to which $G'$ must belong. Here we consider graph parameters instead of graph classes.

Blocker problems are also related to other well-known graph problems as shown for instance in \cite{DPPR15,PPR16}. So far, the literature mainly focused on the following graph parameters:  the chromatic number, the independence number, the clique number, the matching number, the vertex cover number and the domination number. Furthermore, the set $\mathcal{O}$ usually consisted of a single graph operation, namely either vertex deletion, edge contraction, edge deletion or edge addition. Since these blocker problems are usually $\mathsf{NP}$-hard in general graphs, a particular attention has been paid to their computational complexity when restricted to special graph classes.

Recently, the authors in \cite{PAL-GAL-RIE} studied the blocker problem with respect to the domination number and edge contractions. More precisely, they consider the problem of deciding if for a given connected graph $G$ it is possible to obtain a graph $G'$ by contracting at most $k$ edges, where $k\geq 0$ is fixed, such that $\gamma(G')~\leq~\gamma(G) -1$, where $\gamma$ represents the domination number?  For $k=1$, they provided an almost dichotomy (only one family of graphs remained open) for this problem when restricted to $H$-free graphs. Very recently, this last open problem was solved as well (see \cite{TODO}). In this paper, we continue this line of research by considering the total domination number.

More specifically, let $G=(V,E)$ be a graph. The {\it contraction} of an edge $uv\in E$ removes vertices $u$ and $v$ from $G$ and replaces them by a new vertex that is made adjacent to precisely those vertices that were adjacent to $u$ or $v$ in $G$ (without introducing self-loops nor multiple edges). A graph obtained from $G$ by contracting an edge $e$ will be denoted by $G/e$. A set $D\subseteq V$ is called a \textit{dominating set}, if every vertex in $V\setminus D$ has at least one neighbor in $D$. The \textit{domination number} of a graph $G$, denoted by $\gamma(G)$, is the smallest size of a dominating set in $G$. A set $D\subseteq V$ is called a \textit{total dominating set}, if every vertex in $V$ has at least one neighbor in $D$. The \textit{total domination number} of a graph $G$, denoted by $\gamma_t(G)$, is the smallest size of a total dominating set in $G$. We consider the following problem in this paper:

\begin{center}
\begin{boxedminipage}{.99\textwidth}
\textsc{1-Edge Contraction($\gamma_{t}$)}\\[2pt]
\begin{tabular}{ r p{0.8\textwidth}}
\textit{~~~~Instance:} &A connected graph $G=(V,E)$.\\
\textit{Question:} &Does there exist an edge $e\in V$ such that $\gamma_t(G/e)\leq \gamma_t(G)-1$?
\end{tabular}
\end{boxedminipage}
\end{center}

The problem of reducing domination parameters was first considered by Huang and Xu in \cite{HX10}. The authors denote by $ct_{\gamma_t}(G)$ the minimum number of edge contractions required to transform a given graph $G$ into a graph $G'$ such that $\gamma_t (G') \leq \gamma_t (G) -1$. They prove that for a connected graph $G$, we have $ct_{\gamma_t} (G)\leq 3$. In other words, one can always reduce by at least $1$ the total domination number of a connected graph $G$ by using at most $3$ edge contractions (\cite[Theorem 4.3]{HX10}). They also prove the following theorem, which is a crucial result for our work.

\begin{theorem}[\cite{HX10}]\label{theorem:1totcontracdom}
For a connected graph $G$, $ct_{\gamma_t} (G)=1$ if and only if there exists a minimum total dominating $D$ set in $G$ such that the graph induced by $D$ contains a $P_3$.
\end{theorem}
 
As mentioned above, the authors in \cite{PAL-GAL-RIE} considered the domination number, i.e. they considered the problem above but with $\gamma(G)$ instead of $\gamma_t(G)$ denoted by \textsc{1-Edge Contraction($\gamma$)}. In particular they showed that if $H$ is not an induced subgraph of $P_3+pP_2+tK_1$, for $p\geq 1$ and $t\geq 0$ then \textsc{1-Edge Contraction($\gamma$)} is polynomial-time solvable on $H$-free graphs if and only if $H$ is an induced subgraph of $P_5+tK_1$, for $t\geq 0$. Recently, it was shown that the problem can be solved in polynomial time in $H$-free graphs when $H$ is an induced subgraph of $P_3+pP_2+tK_1$, for $p,t\geq 0$ (see \cite{TODO}). Thus, we have the following dichotomy.

\begin{theorem}
\label{thm:dichotomy}
\contracd{} is polynomial-time solvable for $H$-free graphs if and only if $H$  is an induced subgraph of $P_5 + tK_1$ with $t \geq 0$, or $H$ is an induced subgraph of $P_3 + pK_2 + tK_1$ with $p,t\geq 0$.
\end{theorem}

In this paper, we provide a complete dichotomy for \textsc{1-Edge Contraction($\gamma_{t}$)} in $H$-free graphs. Our main result is as follows.

\begin{theorem}
\label{thm:dichotomyTD}
\contractd{} is polynomial-time solvable for $H$-free graphs if and only if $H$ is an induced subgraph of $P_5 + tK_1$ with $t \geq 0$, or $H$ is an induced subgraph of $P_4 + qP_3 + pK_2 + tK_1$ with $q,p,t \geq 0$.
\end{theorem}

It has been shown in \cite{semidom} that the complexities of the problems \textsc{Dominating set} (i.e., given a graph $G$ and an integer $k\geq 0$, does there exist a dominating set of size at most $k$?) and \textsc{Total dominating set} (i.e., given a graph $G$ and an integer $k\geq 0$, does there exist a total dominating set of size at most $k$?) agree in $H$-free graphs for any graph $H$.
The results above show that there are not only hereditary but even monogenic graph classes (i.e. $H$-free graphs for some graph $H$) for which the complexities of the problems \textsc{1-Edge-Contraction($\gamma_t$)} and \textsc{1-Edge-Contraction($\gamma$)} differ.

This paper is organised as follows. In Section \ref{sec-prelim}, we present definitions and notations that are used throughout the paper. Section~\ref{sec-hard} is devoted to the hardness results of \contractd{} while Section~\ref{sec-easy} presents cases when \contractd{} is polynomial-time solvable. In Section~\ref{sec-proof} we put these results together to proof our main result, Theorem \ref{thm:dichotomy}. We conclude the paper by presenting final remarks and future research directions in Section~\ref{sec-conclusion}.


\section{Preliminaries}
\label{sec-prelim}

Throughout this paper, we only consider graphs which are finite, simple and connected. We refer the reader to \cite{Di05} for any terminology and notation not defined here.

For $n \geq 1$, the path and cycle on $n$ vertices are denoted by $P_n$ and $C_n$ respectively. A path on $n$ vertices may also be called an \emph{$n$-path}. The \emph{claw} is the complete bipartite graph with one partition of size one and the other of size three.

Given a graph $G$, we denote by $V(G)$ its vertex set and by $E(G)$ its edge set. For any vertex $v \in V(G)$, the \emph{neighborhood of $v$ in $G$}, denoted by $N_G(v)$ or simply $N(v)$ if it is clear from the context, is the set of vertices adjacent to $v$, that is, $N_G(v) = \set{w\in V(G)\colon\, vw\in E(G)}$; the \emph{closed neighborhood of $v$ in $G$}, denoted by $N_G[v]$ or simply $N[v]$ if it is clear from the context, is the set of vertices adjacent to $v$ together with $v$, that is, $N_G[v] = N_G(v) \cup \{v\}$. For any subset $S \subseteq V(G)$, the \emph{neighborhood of $S$ in $G$}, denoted by $N_G(S)$ or simply $N(S)$ if it is clear from context, is the set $\cup_{v\in S}N(v)$, and the \emph{closed neighborhood of $S$ in $G$}, denoted by $N_G[S]$ or simply $N[S]$ if it is clear from the context, is the set $N_G(S) \cup S$. For an edge $xy \in E(G)$, we may write $N_G(xy)$ (resp. $N_G[xy]$) in place of $N_G(\set{x,y})$ (resp. $N_G[\set{x,y}]$) for simplicity. Similarly, for a family $\mathcal{F} \subseteq E(G)$ of edges, we may write $N_G(\mathcal{F})$ (resp. $N_G[\mathcal{F}]$) in place of $\cup_{e\in\mathcal{F}}N_G(e)$ (resp. $\cup_{e\in\mathcal{F}}N_G[e]$) for simplicity. Let $A,B\subseteq V(G)$. We say that $A$ is complete (resp. anticomplete) to $B$, if every vertex in $A$ is adjacent (resp. non adjacent) to every vertex in $B$. For a subset $S\subseteq V(G)$, we let $G[S]$ denote the subgraph \emph{induced} by $S$, which has vertex set $S$ and edge set $\set{xy\in E(G)\colon\, x,y\in S}$. Given a subset $S \subseteq V(G)$ and a graph $H$, we say that $S$ \emph{contains an (induced) $H$} if $G[S]$ contains $H$ as an (induced) subgraph. The \emph{length} of a path in $G$ is its number of edges. For any two vertices $u,v \in V(G)$, the \emph{distance from $u$ to $v$ in $G$}, denoted by $d_G(u,v)$ or simply $d(u,v)$ if it is clear from the context, is the length of a shortest path from $u$ to $v$ in $G$. Similarly, for any two subset $S,S' \subseteq V(G)$, the \emph{distance from $S$ to $S'$ in $G$}, denoted by $d_G(S,S')$ or simply $d(S,S')$ if it is clear from the context, is the minimum length of a shortest from a vertex in $S$ to a vertex in $S'$, that is, $d_G(S,S') = min_{x \in S,y\in S'} d_G(x,y)$. If $S$ consists of a single vertex, say $S=\{x\}$, we may write $d_G(x,S')$ in place of $d_G(\{x\},S')$ for simplicity. The $k$-subdivision of an edge $uw \in E(G)$ consists in replacing it with a path $uv_1\ldots,v_kw$, where $v_1,\ldots,v_k$ are new vertices. 

A subset $K \subseteq V(G)$ is a \emph{clique} of $G$ if any two vertices of $K$ are adjacent in $G$. A subset $S \subseteq V(G)$ is an \emph{independent set} of $G$ if any two vertices of $S$ are nonadjacent in $G$. Given two subsets $S, S'\subseteq V(G)$, we say that $S$ \emph{dominates} $S'$ if $N(v)\cap S\neq\varnothing$ for every $v\in S'$. Given a vertex $v\in V(G)$ and a subset $S\subseteq V(G)$, we say that the vertex $v$ dominates $S$ if the set $\set{v}$ dominates $S$. If $D$ is a total dominating set of $G$ and $v\in D$, we say that a vertex $w\in V(G)$ is a \emph{private neighbour of $v$ with respect to $D$}, or simply a \emph{private neighbor of $v$} if it is clear from the context, if $N(w)\cap D=\set{v}$. For a subset $S\subseteq V(G)$, we say that a vertex $w\in V(G)\setminus S$ is a \emph{private neighbour of $S$ with respect to $D$}, or simply a \emph{private neighbor of $S$} if it is clear from the context, if $N(w)\cap D\subseteq S$ and $\vert N(w)\cap D\vert=1$. A subset $D \subseteq V(G)$ is a \emph{dominating set} of $G$ if every vertex in $V(G) \setminus D$ is adjacent to at least one vertex in $D$; the \emph{domination number of $G$}, denoted by $\gamma(G)$, is the size of a minimum dominating set of $G$. The \textsc{(Even) Dominating Set} problem takes as input a graph $G$ and an (even) integer $k$, and asks whether $G$ has a dominating set of size at most $k$. The $\mathsf{NP}$-hardness of these two problems follows from \cite{garey}.

For a family $\{H_1,\ldots,H_p\}$ of graphs, we say that $G$ is $\{H_1,\ldots,H_p\}$-free if $G$ contains no induced subgraph isomorphic to a graph in $\{H_1,\ldots,H_p\}$. If $p=1$, we may write $H_1$-free in place of $\{H_1\}$-free for simplicity. The \emph{union} of two simple graphs $G$ and $H$ is the graph $G + H$ with vertex set $V(G) \cup V(H)$ and edge set $E(G) \cup E(H)$. The union of $k$ disjoint copies of $G$ is denoted by $kG$.

For $n\in\mathbb{N}$ we denote by $[n]$ the set $\set{1,\ldots,n}$.


\section{Hardness results}
\label{sec-hard}

In this section, we will present several hardness results regarding \contractd{} with respect to $\mathcal{H}$-free graphs, where $\mathcal{H}$ is a family of at most two graphs.

\begin{theorem}
\label{thm:P6P5P2}
\contractd{} is $\mathsf{NP}$-hard when restricted to $\{P_6,P_5+P_2\}$-free graphs.
\end{theorem}

\begin{proof}
We reduce from {\sc Even Dominating Set} with domination number at least 4. Given an instance $(G,2\ell)$ (with $\gamma(G) \geq 4$) of this problem, we construct an equivalent instance $G'$ of \contractd{} as follows. Let $V(G)=\{v_1,v_2, \ldots, v_n\}$ the vertex set of $G$. The vertex set of $G'$ consists of $2\ell$ vertices $x_1,x_2,\ldots, x_{2\ell}$ and $2\ell + 1$ copies of $V(G)$, denoted by $V_0, V_1, \ldots V_{2\ell}$. For any $ 0 \leq i \leq 2\ell$, we denote the vertices of $V_i$ by $v_1^i,v_2^i,\ldots,v_n^i$. The adjacencies in $G'$ are then defined as follows (see Fig.~\ref{fig:G'gt1}):
\begin{itemize}
\item[$\cdot$] $V_0$ is a clique
\end{itemize}
and for any $1 \leq i \leq 2\ell$, 
\begin{itemize}
\item[$\cdot$] $V_i$ is an independent set;
\item[$\cdot$] for any $1\leq j \leq n$, $v_j^i$ is adjacent to $\{v_k^0,v_k \in N_G[v_j]\}$;
\item[$\cdot$] $x_i$ is adjacent to every vertex in $V_0 \cup V_i$; 
\item[$\cdot$] if $i(\text{mod }2)=1$, then $x_i$ is adjacent to $x_{i+1}$.
\end{itemize}
Note that the fact that for any $1\leq i \leq 2 \ell$ and $1 \leq j \leq n$, $v_j^i$ is adjacent to $\{v_k^0,v_k \in N_G[v_j]\}$ is not made explicit in Fig. \ref{fig:G'gt1} for the sake of readability. We now claim the following.

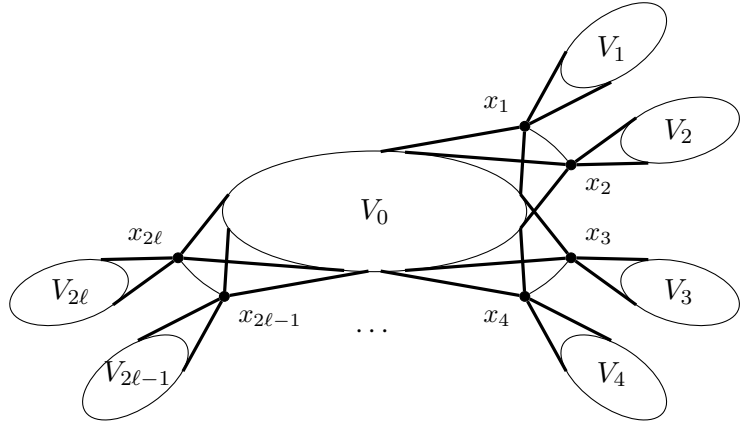
\begin{figure}[htb]
\centering
\begin{tikzpicture}[scale=.8]
\draw (0,0) ellipse (2.5cm and 1cm);
\node[draw=none] at (0,0) {$V_0$};

\node[circ,label=above left:{\small $x_1$}] (x1) at (2.47,1.41) {};
\draw[rotate=35] (4.8,0) ellipse (1cm and .5cm);
\node[draw=none] at (3.9,2.7) {$V_1$};
\draw[very thick] (x1) -- (2.4,.28)
(x1) -- (.1,.99)
(x1) -- (3.15,2.65)
(x1) -- (3.89,2.14);

\node[circ,label=below right:{\small $x_2$}] (x2) at (3.23,0.77) {};
\draw[rotate=15] (5.2,0) ellipse (1cm and .5cm);
\node[draw=none] at (5,1.33) {$V_2$};
\draw[very thick] (x2) -- (2.4,-.28)
(x2) -- (.5,.98)
(x2) -- (4.3,1.55)
(x2) -- (4.5,.8);

\node[circ,label=above right:{\small $x_3$}] (x3) at (3.23,-0.77) {};
\draw[rotate=-15] (5.2,0) ellipse (1cm and .5cm);
\node[draw=none] at (5,-1.33) {$V_3$};
\draw[very thick] (x3) -- (2.4,.28)
(x3) -- (.5,-.98)
(x3) -- (4.3,-1.55)
(x3) -- (4.5,-.8);

\node[circ,label=below left:{\small $x_4$}] (x4) at (2.47,-1.41) {};
\draw[rotate=-35] (4.8,0) ellipse (1cm and .5cm);
\node[draw=none] at (3.9,-2.7) {$V_4$};
\draw[very thick] (x4) -- (2.4,-.28)
(x4) -- (.1,-.99)
(x4) -- (3.15,-2.65)
(x4) -- (3.89,-2.14);

\node[draw=none] at (0,-2) {$\cdots$};

\node[circ,label=below right:{\small $x_{2\ell - 1}$}] (xl-1) at (-2.47,-1.41) {};
\draw[rotate=35] (-4.8,0) ellipse (1cm and .5cm);
\node[draw=none] at (-3.9,-2.7) {$V_{2\ell - 1}$};
\draw[very thick] (xl-1) -- (-2.4,-.28)
(xl-1) -- (-.1,-.99)
(xl-1) -- (-3.15,-2.65)
(xl-1) -- (-3.89,-2.14);

\node[circ,label=above left:{\small $x_{2\ell}$}] (xl) at (-3.23,-0.77) {};
\draw[rotate=15] (-5.2,0) ellipse (1cm and .5cm);
\node[draw=none] at (-5,-1.33) {$V_{2\ell}$};
\draw[very thick] (xl) -- (-2.4,.28)
(xl) -- (-.5,-.98)
(xl) -- (-4.3,-1.55)
(xl) -- (-4.5,-.8);

\draw (x1) edge[bend left=10] (x2);
\draw (x3) edge[bend left=10] (x4);
\draw (xl-1) edge[bend left=10] (xl);
\end{tikzpicture}
\caption{The graph $G'$ (thick lines indicate that the vertex $x_i$ is adjacent to every vertex of $V_0 \cup V_i$ for any $1 \leq i \leq 2\ell$).}
\label{fig:G'gt1}
\end{figure}

\begin{claim}
\label{clm:gammat}
$\gamma_t(G') = \min \{ \gamma (G), 2 \ell\}$.
\end{claim}

\begin{claimproof}
It is clear that $\{x_1,x_2,\ldots, x_{2\ell -1},x_{2\ell}\}$ is a total dominating set of $G'$; thus, $\gamma_t (G') \leq 2 \ell$. If $\gamma (G) \leq 2 \ell$ and $\{v_{i_1},v_{i_2},\ldots, v_{i_k}\}$ is a minimum dominating set of $G$, it follows from the construction above that $\{v_{i_1}^0,v_{i_2}^0,\ldots, v_{i_k}^0\}$ is a total dominating set of $G'$ (recall that $V_0$ is a clique). Thus, $\gamma_t(G') \leq \gamma (G)$ and so, $\gamma_t(G') \leq \min \{ \gamma (G), 2 \ell \}$. Now suppose that $\gamma_t(G') < \min \{ \gamma (G), 2 \ell \}$ and consider a minimum total dominating set $D$ of $G'$. Then there must exist $i \in \{1,\ldots, 2 \ell\}$ such that $x_i \not\in D$; indeed, if for all $i \in \{1,\ldots, 2 \ell\}$, $x_i \in D$ then we would have that $\vert D \vert \geq 2 \ell$, thereby contradicting the fact that $\gamma_t(G') < \min \{ \gamma (G), 2 \ell \}$. But then, $D' = D \cap (V_0 \cup V_i)$ must dominate every vertex in $V_i$ and so, $\vert D' \vert \geq \gamma (G)$. But $\vert D' \vert \leq \vert D \vert $ which implies that $\gamma (G) \leq \vert D \vert$, a contradiction. Therefore, $\gamma_t(G') = \min \{ \gamma (G), 2 \ell \}$.
\end{claimproof}

We now show that $(G,2\ell)$ (with $\gamma(G) \geq 4$) is a \yes-instance for {\sc Even Dominating Set} if and only if $G'$ is a \yes-instance for \contractd{}.\\

First assume that $\gamma (G) \leq 2 \ell$. Then by Claim \ref{clm:gammat}, $\gamma_t(G') = \gamma (G)$ and if $\{v_{i_1},v_{i_2},\ldots,v_{i_k}\}$ is a minimum dominating set of $G$ then $\{v_{i_1}^0,v_{i_2}^0,\ldots,v_{i_k}^0\}$ is a minimum total dominating set of $G'$ containing a $P_3$ (recall that $V_0$ is a clique and $\gamma(G) \geq 4$). We then conclude by Theorem~\ref{theorem:1totcontracdom} that $G'$ is a \yes-instance for \contractd{}.

Conversely, assume that $G'$ is a \yes-instance for \contractd{}, that is, there exists a minimum total dominating set $D$ of $G'$ containing a $P_3$ (see Theorem~\ref{theorem:1totcontracdom}). Then there must exist $i \in \{1,\ldots,2\ell\}$ such that $x_i \not\in D$; indeed, if for all $i \in \{1,\ldots,2\ell\}$, $x_i \in D$ then $\vert D \vert \geq 2 \ell$ and we conclude by Claim \ref{clm:gammat} that in fact equality holds. It follows that $D$ consists of $x_1,x_2,\ldots,x_{2\ell-1},x_{2\ell}$; in particular, $D$ contains no $P_3$, a contradiction. Thus, there exists $1 \leq i \leq 2\ell$ such that $x_i \notin D$ and so, $D' = D \cap (V_0 \cup V_i)$ must dominate every vertex in $V_i$. It follows that $\vert D' \vert \geq \gamma (G)$ and since $\vert D' \vert  \leq \vert D \vert$, we conclude that $\gamma (G) \leq \vert D \vert \leq 2 \ell$ by Claim~\ref{clm:gammat}, that is, $(G,2\ell)$ is a \yes-instance for {\sc Even Dominating Set}.\\    

Finally, it is easy to see that $G'$ is $P_6$-free as well as $(P_5+P_2)$-free which concludes the proof.
\end{proof}

\begin{theorem}
\label{thm:2P4} 
\contractd\ is $\mathsf{coNP}$-hard when restricted to $2P_4$-free graphs.
\end{theorem}

\begin{proof}
We reduce from \textsc{3-Sat} as follows. Given an instance $\Phi$ of this problem, with variable set $X$ and clause set $C$, we construct a graph $G_{\Phi}$ such that $\Phi$ is satisfiable if and only if $G_{\Phi}$ is a \no-instance for \contractd, as follows. For any variable $x \in X$, we introduce the gadget $G_x$ depicted in Figure~\ref{fig:2P4vargad} with one distinguished \textit{positive literal vertex $x$} and one distinguished \textit{negative literal vertex $\bar{x}$}. For any clause $c \in C$, we introduce a \textit{clause vertex $c$} which is made adjacent to the (positive or negative) literal vertices whose corresponding literal occurs in $c$. Finally, we add an edge between any two clause vertices so that the set of clause vertices induces a clique denoted by $K$ in the following. 

\begin{figure}[htb]
\centering
\begin{tikzpicture}
\node[cir,label=below:{\small $x$}] (x) at (0,0) {};
\node[cir,label=below:{\small $\bar{x}$}] (barx) at (1,0) {};
\node[cir,label=right:{\small $u_x$}] (ux) at (0.5,1) {};
\node[cir,label=right:{\small $v_x$}] (vx) at (.5,2) {};

\draw[-] (x) -- (barx)
(x) -- (ux)
(barx) -- (ux)
(ux) -- (vx);
\end{tikzpicture}
\caption{The variable gadget $G_x$.}
\label{fig:2P4vargad}
\end{figure}
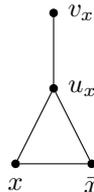

\begin{observation}
\label{obs:vargad}
For any total dominating set $D$ of $G_{\Phi}$ and any variable $x \in X$, $\vert D \cap V(G_x) \vert \geq 2$ and $u_x \in D$. In particular, $\gamma_t(G_{\Phi}) \geq 2 \vert X \vert$.
\end{observation} 	

Indeed, since $v_x$ should be dominated, necessarily $u_x \in D$ and since $u_x$ should be dominated, $D \cap \{v_x,x,\bar{x}\} \neq \varnothing$.\\

\begin{claim}
\label{clm:2P4phisat}
$\Phi$ is satisfiable if and only if $\gamma_t(G_{\Phi}) = 2 \vert X \vert$.
\end{claim}

\begin{claimproof}
Assume that $\Phi$ is satisfiable and consider a truth assignment satisfying $\Phi$. We construct a total dominating set of $G_{\Phi}$ as follows. For any variable $x \in X$, if $x$ is true then we add $x$ and $u_x$ to $D$; otherwise, we add $\bar{x}$ and $u_x$ to $D$. Clearly, $D$ is a total dominating set as every clause is satisfied and we conclude by Observation~\ref{obs:vargad} that $D$ is minimum.

Conversely, assume that $\gamma_t(G_{\Phi}) = 2 \vert X \vert$ and consider a minimum total dominating set $D$ of $G_{\Phi}$. First observe that by Observation~\ref{obs:vargad}, $\vert D \cap V(G_x) \vert = 2$ and $\vert D \cap \{x,\bar{x}\} \vert \leq 1$ for any $x \in X$, which implies in particular that $D \cap K = \varnothing$. It follows that for any clause vertex $c$, there must exist $x \in X$ such that the (positive or negative) literal vertex whose corresponding literal occurs in $c$ belongs to $D$. We may thus construct a truth assignment satisfying $\Phi$ as follows. For any variable $x \in X$, if the positive literal vertex $x$ belongs to $D$ then we set $x$ to true; if the negative literal vertex $\bar{x}$ belongs to $D$ then we set $x$ to false; otherwise, we set $x$ to true.
\end{claimproof}

\begin{claim}
\label{clm:2P4contractd}
$\gamma_t(G_{\Phi}) = 2 \vert X \vert$ if and only if $G_{\Phi}$ is a \no-instance for \contractd.
\end{claim}

\begin{claimproof}
Assume that $\gamma_t(G_{\Phi}) = 2 \vert X \vert$ and let $D$ be a minimum total dominating set of $G_{\Phi}$. Then by Observation~\ref{obs:vargad}, $\vert D \cap V(G_x) \vert = 2$ for any $x \in X$ which implies in particular that $D \cap K = \varnothing$. But then, it is clear that $D$ contains no $P_3$, and hence $G_{\Phi}$ is a \no-instance for \contractd according to Theorem \ref{theorem:1totcontracdom}.

Conversely, assume that $G_{\Phi}$ is a \no-instance for \contractd\ and consider a minimum total dominating set $D$ of $G_{\Phi}$. First observe that since $D$ contains no $P_3$ (see Theorem \ref{theorem:1totcontracdom}), necessarily $\vert D \cap V(G_x) \vert \leq 2$ for any $x \in X$; we then conclude by Observation~\ref{obs:vargad} that in fact equality holds for any $x \in X$. We now claim that $D \cap K = \varnothing$. Indeed, suppose to the contrary that there exists a clause vertex $c$ such that $c \in D$ and consider a variable $x$ occuring in $c$, say $x$ occurs positive in $c$ without loss of generality. Then $(D \setminus \{v_x,\bar{x}\}) \cup \{x\}$ is a minimum total dominating set containing a $P_3$, a contradiction. Thus, $D \cap K = \varnothing$ and so, $\vert D \vert = 2 \vert X \vert$.
\end{claimproof}

Now by combining Claims~\ref{clm:2P4phisat} and \ref{clm:2P4contractd}, we obtain that $\Phi$ is satisfiable if and only if $G_{\Phi}$ is a \no-instance for \contractd. Since $G_{\Phi}$ is obviously $2P_4$-free, this concludes the proof.
\end{proof}

\begin{theorem}
\label{thm:tdclaw}
\contractd is $\mathsf{coNP}$-hard when restricted to claw-free graphs.
\end{theorem}

\begin{proof}
We reduce from \textsc{Positive Cubic 1-In-3 3-Sat} which was shown to be $\mathsf{NP}$-hard in~\cite{1IN3}. It is a variant of the \textsc{3-Sat} problem where each variable occurs only nonnegated and in exactly three clauses, and the formula is satisfiable if and only if there exists a truth assignment to the variables such that each clause has exactly one true literal. Given an instance $\Phi$ of this problem, with variable set $X$ and clause set $C$, we contruct a graph $G_{\Phi}$ such that $\Phi$ is satisfiable if and only if $G_{\Phi}$ is a \no-instance for \contractd, as follows. For each variable $x \in X$ contained in clauses $c,c'$ and $c''$, we introduce the gadget $G_x$ depicted in Figure~\ref{fig:vargadget}. For each clause $c \in C$ containing variables $x,y$ and $z$, we introduce the gadget $G_c$ which is the disjoint union of the graphs $G_c^T$ and $G_c^F$ depicted in Figure~\ref{fig:clausegadget}; then for all $\ell \in \{x,y,z\}$, we add an edge between $t^\ell_c$ and $t^c_\ell$, and $f^\ell_c$ and $f^c_\ell$.

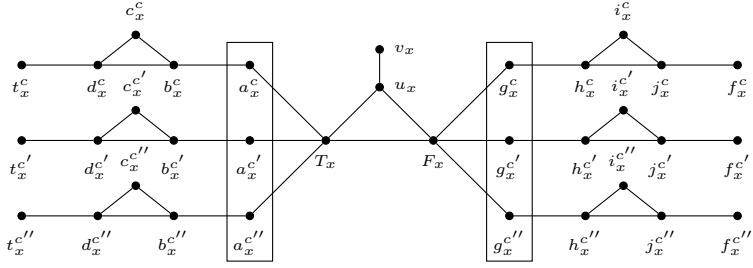
\begin{figure}[htb]
\centering
\begin{tikzpicture}[node distance=1cm]
\node[cir,label=right:{\tiny $u_x$}] (ux) at (0,0) {};
\node[cir,label=right:{\tiny $v_x$}] at ($(ux) + (0,.5)$) (vx) {};
\node[cir,below left of=ux,label=below:{\tiny $T_x$}] (Tx) {};
\node[cir,below right of=ux,label=below:{\tiny $F_x$}] (Fx) {};

\node[cir,left of=Tx,label=below:{\tiny $a^{c'}_x$}] (b1) {};
\node[cir,above of=b1,label=below:{\tiny $a^c_x$}] (a1) {};
\node[cir,below of=b1,label=below:{\tiny $a^{c''}_x$}] (c1) {};

\node[cir,left of=a1,label=below:{\tiny $b^c_x$}] (d1) {};
\node[cir,left of=d1,label=below:{\tiny $d^c_x$}] (j1) {};
\node[cir,label=above:{\tiny $c^c_x$}] at ($(j1) + (.5,.4)$) (g1) {};
\node[cir,left of=j1,label=below:{\tiny $t^c_x$}] (t1x) {};

\node[cir,left of=b1,label=below:{\tiny $b^{c'}_x$}] (e1) {};
\node[cir,left of=e1,label=below:{\tiny $d^{c'}_x$}] (k1) {};
\node[cir,label=above:{\tiny $c^{c'}_x$}] at ($(k1) + (.5,.4)$) (h1) {};
\node[cir,left of=k1,label=below:{\tiny $t^{c'}_x$}] (t2x) {};

\node[cir,left of=c1,label=below:{\tiny $b^{c''}_x$}] (f1) {};
\node[cir,left of=f1,label=below:{\tiny $d^{c''}_x$}] (l1) {};
\node[cir,label=above:{\tiny $c^{c''}_x$}] at ($(l1) + (.5,.4)$) (i1) {};
\node[cir,left of=l1,label=below:{\tiny $t^{c''}_x$}] (t3x) {};

\node[cir,right of=Fx,label=below:{\tiny $g^{c'}_x$}] (b2) {};
\node[cir,above of=b2,label=below:{\tiny $g^c_x$}] (a2) {};
\node[cir,below of=b2,label=below:{\tiny $g^{c''}_x$}] (c2) {};

\node[cir,right of=a2,label=below:{\tiny $h^c_x$}] (d2) {};
\node[cir,right of=d2,label=below:{\tiny $j^c_x$}] (j2) {};
\node[cir,label=above:{\tiny $i^c_x$}] at ($(d2) + (.5,.4)$) (g2) {};
\node[cir,right of=j2,label=below:{\tiny $f^c_x$}] (f1x) {};

\node[cir,right of=b2,label=below:{\tiny $h^{c'}_x$}] (e2) {};
\node[cir,right of=e2,label=below:{\tiny $j^{c'}_x$}] (k2) {};
\node[cir,label=above:{\tiny $i^{c'}_x$}] at ($(e2) + (.5,.4)$) (h2) {};
\node[cir,right of=k2,label=below:{\tiny $f^{c'}_x$}] (f2x) {};

\node[cir,right of=c2,label=below:{\tiny $h^{c''}_x$}] (f2) {};
\node[cir,right of=f2,label=below:{\tiny $j^{c''}_x$}] (l2) {};
\node[cir,label=above:{\tiny $i^{c''}_x$}] at ($(f2) + (.5,.4)$) (i2) {};
\node[cir,right of=l2,label=below:{\tiny $f^{c''}_x$}] (f3x) {};

\draw[-] (ux) -- (vx)
(ux) -- (Tx) 
(ux) -- (Fx)
(Tx) -- (Fx)
(Tx) -- (a1)
(Tx) -- (b1)
(Tx) -- (c1)
(a1) -- (d1)
(d1) -- (g1)
(d1) -- (j1)
(g1) -- (j1)
(j1) -- (t1x)
(b1) -- (e1)
(e1) -- (h1)
(e1) -- (k1)
(h1) -- (k1)
(k1) -- (t2x)
(c1) -- (f1)
(f1) -- (i1)
(f1) -- (l1)
(i1) -- (l1)
(l1) -- (t3x) 
(Fx) -- (a2)
(Fx) -- (b2)
(Fx) -- (c2)
(a2) -- (d2)
(d2) -- (g2)
(d2) -- (j2)
(g2) -- (j2)
(j2) -- (f1x)
(b2) -- (e2)
(e2) -- (h2)
(e2) -- (k2)
(h2) -- (k2)
(k2) -- (f2x)
(c2) -- (f2)
(f2) -- (i2)
(f2) -- (l2)
(i2) -- (l2)
(l2) -- (f3x); 

\draw ($(c1) + (-.3,-.6)$) rectangle ($(a1) + (.3,.3)$);
\draw ($(c2) + (-.3,-.6)$) rectangle ($(a2) + (.3,.3)$);
\end{tikzpicture}
\caption{The variable gadget $G_x$ for a variable $x$ contained in clauses $c,c'$ and $c''$ (a rectangle indicates that the corresponding set of vertices induces a clique).}
\label{fig:vargadget}
\end{figure}

\begin{figure}[htb]
\centering
\begin{subfigure}[b]{.45\textwidth}
\centering
\begin{tikzpicture}[node distance=1cm]
\node[cir,label=below:{\tiny $u_c$}] (u) at (0,0) {};
\node[cir,right of=u,label=below:{\tiny $a^y_c$}] (b) {};
\node[cir,above of=b,label=below:{\tiny $a^x_c$}] (a) {};
\node[cir,below of=b,label=below:{\tiny $a^z_c$}] (c) {};

\node[cir,right of=a,label=below:{\tiny $c^x_c$}] (g1) {};
\node[cir,label=above:{\tiny $b^x_c$}] at ($(a) + (.5,.4)$) (d1) {};
\node[cir,right of=g1,label=below:{\tiny $d^x_c$}] (j1) {};
\node[cir,right of=j1,label=below:{\tiny $t^x_c$}] (txc) {};

\node[cir,right of=b,label=below:{\tiny $c^y_c$}] (g2) {};
\node[cir,label=above:{\tiny $b^y_c$}] at ($(b) + (.5,.4)$) (d2) {};
\node[cir,right of=g2,label=below:{\tiny $d^y_c$}] (j2) {};
\node[cir,right of=j2,label=below:{\tiny $t^y_c$}] (tyc) {};

\node[cir,right of=c,label=below:{\tiny $c^z_c$}] (g3) {};
\node[cir,label=above:{\tiny $b^z_c$}] at ($(c) + (.5,.4)$) (d3) {};
\node[cir,right of=g3,label=below:{\tiny $d^z_c$}] (j3) {};
\node[cir,right of=j3,label=below:{\tiny $t^z_c$}] (tzc) {};

\draw[-] (u) -- (a)
(u) -- (b)
(u) -- (c)
(a) -- (d1)
(a) -- (g1)
(d1) -- (g1)
(g1) -- (j1)
(j1) -- (txc)
(b) -- (d2)
(b) -- (g2)
(d2) -- (g2)
(g2) -- (j2)
(j2) -- (tyc)
(c) -- (d3)
(c) -- (g3)
(d3) -- (g3)
(g3) -- (j3)
(j3) -- (tzc);

\draw ($(c) + (-.2,-.5)$) rectangle ($(a) + (.2,.2)$);
\end{tikzpicture}
\caption{The graph $G^T_c$ (the rectangle indicates that the corresponding set of vertices induces a clique).}
\end{subfigure}
\hspace*{.5cm}
\begin{subfigure}[b]{.45\textwidth}
\centering
\begin{tikzpicture}[node distance=1cm]
\node[cir,label=below:{\tiny $v_c$}] (vc) at (0,0) {};
\node[cir,right of=vc,label=below:{\tiny $w_c$}] (a) {};
\node[cir,left of=vc,label=below:{\tiny $g^y_c$}] (c) {};
\node[cir,above of=c,label=below:{\tiny $g^x_c$}] (b) {};
\node[cir,below of=c,label=below:{\tiny $g^z_c$}] (d) {};
\node[cir,left of=b,label=below:{\tiny $f^x_c$}] (fxc) {};
\node[cir,left of=c,label=below:{\tiny $f^y_c$}] (fyc) {};
\node[cir,left of=d,label=below:{\tiny $f^z_c$}] (fzc) {};

\draw[-] (vc) -- (a)
(vc) -- (b)
(vc) -- (c)
(vc) -- (d)
(b) -- (fxc)
(c) -- (fyc)
(d) -- (fzc);

\draw ($(d) + (-.2,-.5)$) rectangle ($(b) + (.2,.2)$);
\end{tikzpicture}
\caption{The graph $G^F_c$ (the rectangle indicates that the corresponding set of vertices induces a clique).}
\end{subfigure}
\caption{The clause gadget $G_c$ is the disjoint union of $G^T_c$ and $G^F_c$ for a clause $c$ containing variables $x,y$ and $z$.}
\label{fig:clausegadget}
\end{figure}
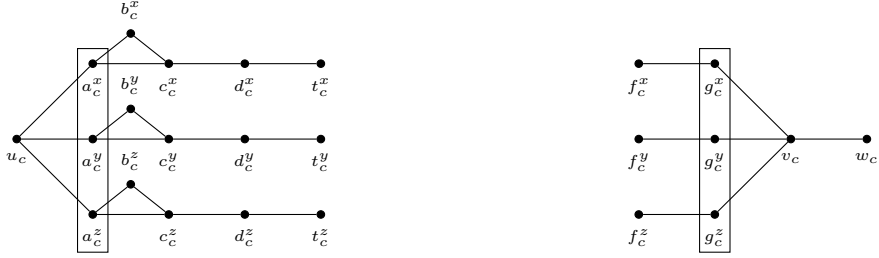

\noindent
We begin with the following easy observations.

\begin{observation}
\label{obs:clause}
Let $D$ be a total dominating set of $G_{\Phi}$. Then for each clause $c \in C$ with variables $x,y$ and $z$, the following holds.
\begin{itemize}
\item[(i)] $\vert D \cap \{g^x_c,g^y_c,g^z_c,v_c,w_c\} \vert \geq 2$ and $v_c \in D$.
\item[(ii)] For any $\ell \in \{x,y,z\}$, $D \cap \{a_c^\ell,c_c^\ell\} \neq \varnothing$ and $\vert D \cap \{a^\ell_c,b^\ell_c,c^\ell_c,d^\ell_c,t^\ell_c\} \vert \geq 2$.
\end{itemize}
In particular, $\vert D \cap V(G_c) \vert \geq 8$.
\end{observation}

(i) Indeed, since $w_c$ must be dominated, necessarily $v_c \in D$ and since $v_c$ must be dominated, $D \cap  \{g^x_c,g^y_c,g^z_c,w_c\} \neq \varnothing$.

(ii) Indeed, since $c^\ell_c$ must be dominated, either $a^\ell_c \in D$ or $c^\ell_c \in D$. If $a^\ell_c \in D$ then $D \cap \{c^\ell_c,t^\ell_c\} \neq \varnothing$ as $d^\ell_c$ should be dominated; and if $c^\ell_c \in D$ then $D \cap \{a^\ell_c,b^\ell_c,d^\ell_c\} \neq \varnothing$ as $c^\ell_c$ should be dominated. $\diamond$\\

\begin{observation}
\label{obs:variable}
Let $D$ be a total dominating set of $G_{\Phi}$. Then for each variable $x \in X$ contained in clauses $c,c'$ and $c''$, the following holds.
\begin{itemize}
\item[(i)] $\vert D \cap \{u_x,v_x,T_x,F_x\} \vert \geq 2$ and $u_x \in D$.
\item[(ii)] For any $\ell \in \{c,c',c''\}$, $D \cap \{b^\ell_x,d^\ell_x\} \neq \varnothing$ and $\vert D \cap \{a^\ell_x,b^\ell_x,c^\ell_x,d^\ell_x,t^\ell_x\} \vert \geq 2$.
\item[(iii)] for any $\ell \in \{c,c',c''\}$, $D \cap \{h^\ell_x,j^\ell_x\} \neq \varnothing$ and $\vert D \cap \{g^\ell_x,h^\ell_x,i^\ell_x,j^\ell_x,f^\ell_x\} \vert \geq 2$.
\end{itemize}
In particular, $\vert D \cap V(G_x) \vert \geq 14$.
\end{observation}

(i) Indeed, since $v_x$ must be dominated, necessarily $u_x \in D$ and since $u_x$ must be dominated, $D \cap \{v_x,T_x,F_x\} \neq \varnothing$.

(ii) Indeed, since $c^\ell_x$ must be dominated, either $b^\ell_x \in D$ or $d^\ell_x \in D$. If $b^\ell_x \in D$ then $D \cap \{a^\ell_x,c^\ell_x,d^\ell_x\} \neq \varnothing$ as $b^\ell_x$ should be dominated; and if $d^\ell_x \in D$ then $D \cap \{b^\ell_x,c^\ell_x,t^\ell_x\} \neq \varnothing$ as $d^\ell_x$ should be dominated. The proof for (iii) is symmetric.~$\diamond$\\

\noindent
We now prove the following two claims.

\begin{claim}
\label{clm:phisat}
$\gamma_t(G_{\Phi}) = 14 \vert X \vert + 8 \vert C \vert$ if and only if $\Phi$ is satisfiable.
\end{claim}

\begin{claimproof}
Assume first that $\Phi$ is satisfiable and consider a truth assignment satisfying $\Phi$. We construct a total dominating set $D$ for $G_{\Phi}$ as follows. For any variable $x \in X$ contained in clauses $c,c'$ and $c''$, if $x$ is true then add $\{u_x,T_x,d^c_x,d^{c'}_x,d^{c''}_x,t^c_x,t^{c'}_x,t^{c''}_x,h^c_x,h^{c'}_x,h^{c''}_x,j^c_x,j^{c'}_x,j^{c''}_x\}$ to $D$; otherwise add $\{u_x,F_x,j^c_x,j^{c'}_x,j^{c''}_x,f^c_x,f^{c'}_x,f^{c''}_x,b^c_x,b^{c'}_x,b^{c''}_x,d^c_x,d^{c'}_x,d^{c''}_x\}$ to $D$. For any clause $c \in C$ containing variables $x,y$ and $z$, we proceed as follows. Assume without loss of generality that $x$ is true (and thus $y$ and $z$ are false). Then add $\{c^x_c,a^x_c,d^y_c,c^y_c,d^z_c,c^z_c,g^x_c,v_c\}$ to $D$. Clearly, $D$ is a total dominating set and we conclude by Observations~\ref{obs:clause} and \ref{obs:variable} that $D$ is minimum. Thus, $\gamma_t(G_{\Phi}) = 14 \vert X \vert + 8 \vert C \vert$.\\

Conversely, assume that $\gamma_t(G_{\Phi}) = 14 \vert X \vert + 8 \vert C \vert$. Let us first make several observations. The following is a straightforward consequence of Observation~\ref{obs:clause}.

\begin{observation}
\label{obs:clause2}
Let $D$ be a total dominating set of $G_{\Phi}$. Then for any clause $c \in C$ containing variables $x,y$ and $z$, if $\vert D \cap V(G_c) \vert = 8$ then $D \cap \{f_c^x,f_c^y,f_c^z,u_c\} = \varnothing$.
\end{observation}

\begin{observation}
\label{obs:variable2}
Let $D$ be a total dominating set of $G_{\Phi}$. Then for any variable $x \in X$ contained in $c,c',$ and $c''$, if $\vert D \cap V(G_x) \vert = 14$, the following holds.
\begin{itemize}
\item[(i)] For any $\ell \in \{c,c',c''\}$, if $t_x^\ell \in D$ then $D \cap \{a_x^\ell,b_x^\ell,c_x^\ell,d_x^\ell,t_x^\ell\} = \{d_x^\ell,t_x^\ell\}$.
\item[(ii)] For any $\ell \in \{c,c',c''\}$, if $a_x^\ell \in D$ then $D \cap \{a_x^\ell,b_x^\ell,c_x^\ell,d_x^\ell,t_x^\ell\} = \{a_x^\ell,b_x^\ell\}$.
\item[(iii)] For any $\ell \in \{c,c',c''\}$, if $f_x^\ell \in D$ then $D \cap \{g_x^\ell,h_x^\ell,i_x^\ell,j_x^\ell,f_x^\ell\} = \{j_x^\ell,f_x^\ell\}$.
\item[(iv)] For any $\ell \in \{c,c',c''\}$, if $g_x^\ell \in D$ then $D \cap \{g_x^\ell,h_x^\ell,i_x^\ell,j_x^\ell,f_x^\ell\} = \{g_x^\ell,h_x^\ell\}$.
\end{itemize}
\end{observation}

(i) Indeed, note first that since $\vert D \cap V(G_x) \vert = 14$, we have by Observation \ref{obs:variable} that $\vert D \cap \{a_x^\ell,b_x^\ell,c_x^\ell,d_x^\ell,t_x^\ell\} \vert = 2$ for all $\ell \in\{c,c',c''\}$. Thus, if $t_x^\ell \in D$ then by Observation~\ref{obs:variable}(ii), $\vert D \cap \{b_x^\ell,d_x^\ell\} \vert = 1$ (note that in particular, $a_x^\ell \notin D$); but if $b_x^\ell \in D$ then $b_x^\ell$ is not dominated as $a_x^\ell \notin D$. Therefore, if $t_x^\ell \in D$ then $d_x^\ell \in D$. The proof for (iii) is symmetric.

(ii) Similarly, if $a_x^\ell \in D$ then by Observation~\ref{obs:variable}(ii), $\vert D \cap \{b_x^\ell,d_x^\ell\} \vert = 1$ (note that in particular, $t_x^\ell \notin D$); but if $d_x^\ell \in D$ then $d_x^\ell$ is not dominated as $t_x^\ell \notin D$. Therefore, if $a_x^\ell \in D$ then $b_x^\ell \in D$. The proof for (iv) if symmetric.~$\diamond$\\

\begin{observation}
\label{obs:FxTx}
Let $D$ be a total dominating set of $G_{\Phi}$. For any variable $x \in X$ contained in clauses $c,c'$ and $c''$, if $\vert D \cap V(G_x) \vert = 14$ and $\vert D \cap V(G_\ell) \vert = 8$ for all $\ell \in \{c,c',c''\}$, then the following holds.
\begin{itemize}
\item[(i)] If there exists $\ell \in \{c,c',c''\}$ such that $t_x^\ell \in D$ then $T_x \in D$.
\item[(ii)] If there exists $\ell \in \{c,c',c''\}$ such that $f_x^\ell \in D$ then $F_x \in D$. 
\end{itemize}
\end{observation}

(i) Indeed, note first that since $\vert D \cap V(G_x) \vert = 14$, we have by Observation \ref{obs:variable} that $\vert D \cap \{a_x^\ell,b_x^\ell,c_x^\ell,d_x^\ell,t_x^\ell\} \vert = 2$ for all $\ell \in\{c,c',c''\}$. Similarly, since $\vert D \cap V(G_\ell) \vert = 8$ for all $\ell \in \{c,c',c''\}$, it follows from Observation~\ref{obs:clause} that $\vert D \cap \{a_\ell^x,b_\ell^x,c_\ell^x,d_\ell^x,t_\ell^x\} \vert = 2$ for all $\ell \in \{c,c',c''\}$. Now assume that there exists $\ell \in \{c,c',c''\}$ such that $t_x^\ell \in D$, say $c$ without loss of generality, and suppose to the contrary that $T_x \notin D$. Then by Observation~\ref{obs:variable2}(i), $D \cap \{a_x^c,b_x^c,c_x^c,d_x^c,t_x^c\} = \{d_x^c,t_x^c\}$. Thus since $a_x^c$ should be dominated and $T_x \notin D$, there must exist $p \in \{c',c''\}$ such that $a_x^p \in D$, say $c'$ without loss of generality. But then by Observation~\ref{obs:variable2}(ii), $D \cap \{a_x^{c'},b_x^{c'},c_x^{c'},d_x^{c'},t_x^{c'}\} = \{a_x^{c'},b_x^{c'}\}$ and so, $t_{c'}^x \in D$ for otherwise $t_x^{c'}$ would not be dominated. But $\vert D \cap \{a_{c'}^x,b_{c'}^x,c_{c'}^x,d_{c'}^x,t_{c'}^x\} \vert = 2$ and $D \cap \{a_{c'}^x,c_{c'}^x\} \neq \varnothing$ by Observation~\ref{obs:clause}, which implies that $d_{c'}^x \notin D$ and so, $t_{c'}^x$ is not dominated, a contradiction. Thus, $T_x \in D$.

(ii) Assume that there exists $\ell \in \{c,c',c''\}$ such that $f_x^\ell \in D$, say $c$ without loss of generality, and suppose to the contrary that $F_x \notin D$. Then by Observation~\ref{obs:variable2}(iii), $D \cap \{g_x^c,h_x^c,i_x^c,j_x^c,f_x^c\} = \{j_x^c,f_x^c\}$. Thus since $g_x^c$ should be dominated and $F_x \notin D$, there must exist $p \in \{c',c''\}$ such that $g_x^p \in D$, say $c'$ without loss of generality. But then by Observation~\ref{obs:variable2}(iv), $D \cap \{g_x^{c'},h_x^{c'},i_x^{c'},j_x^{c'},f_x^{c'}\} = \{g_x^{c'},h_x^{c'}\}$ and so, $f_{c'}^x$ must belong to $D$ ($f_x^{c'}$ would otherwise not be dominated) which contradicts Observation~\ref{obs:clause2} (recall that $\vert D \cap V(G_{c'}) \vert = 8$). Thus, $F_x \in D$.~$\diamond$\\

\begin{remark} 
\label{rem:Gphi}
If $\gamma_t(G_{\Phi}) = 14 \vert X \vert + 8 \vert C \vert$ and $D$ is a minimum total dominating set of $G_{\Phi}$, then the following hold. For any clause $c \in C$ containing variable $x,y$ and $z$, we have by Observation~\ref{obs:clause} that
\begin{itemize}
\item[(i)] $\vert D \cap \{g^x_c,g^y_c,g^z_c,v_c,w_c\} \vert = 2$ and $v_c \in D$; and
\item[(ii)] for any $\ell \in \{x,y,z\}$, $D \cap \{a_c^\ell,c_c^\ell\} \neq \varnothing$ and $\vert D \cap \{a^\ell_c,b^\ell_c,c^\ell_c,d^\ell_c,t^\ell_c\} \vert = 2$.
\end{itemize}
Similarly by Observation~\ref{obs:variable}, we have that for any variable $x \in X$ contained in clauses $c,c'$ and $c''$,  
\begin{itemize}
\item[(i)] $\vert D \cap \{u_x,v_x,T_x,F_x\} \vert = 2$ and $u_x \in D$;
\item[(ii)] for any $\ell \in \{c,c',c''\}$, $D \cap \{b^\ell_x,d^\ell_x\} \neq \varnothing$ and $\vert D \cap \{a^\ell_x,b^\ell_x,c^\ell_x,d^\ell_x,t^\ell_x\} \vert = 2$; and
\item[(iii)] for any $\ell \in \{c,c',c''\}$, $D \cap \{h^\ell_x,j^\ell_x\} \neq \varnothing$ and $\vert D \cap \{g^\ell_x,h^\ell_x,i^\ell_x,j^\ell_x,f^\ell_x\} \vert = 2$.
\end{itemize}
\end{remark}

Turning back to the proof of Claim~\ref{clm:phisat}, let $D$ be a minimum total dominating set of $G_{\Phi}$. We claim the following.
\begin{observation}
\label{obs:sat}
If $\gamma_t(G_{\Phi}) = 14 \vert X \vert + 8 \vert C \vert$ then for any minimum total dominating set $D'$ and any clause $c \in C$ containing variables $x,y$ and $z$, there exists $\ell \in \{x,y,z\}$ such that $t_\ell^c \in D'$ and for any $p \in \{x,y,z\} \setminus \{\ell\}$, $f_p^c \in D'$. 
\end{observation} 

Indeed, suppose to the contrary that for all $\ell \in \{x,y,z\}$, $t_\ell^c \notin D'$. Then $\{d_c^x,d_c^y,d_c^z\} \subset D$ as $t_c^x, t_c^y$ and $t_c^z$ should be dominated, and so $D \cap \{t_c^x, t_c^y,t_c^z\} = \varnothing$ as $D \cap \{a_c^\ell,c_c^\ell\} \neq \varnothing$ and $\vert D \cap \{a^\ell_c,b^\ell_c,c^\ell_c,d^\ell_c,t^\ell_c\} \vert = 2$ for any $\ell \in \{x,y,z\}$. But then by Observation~\ref{obs:clause}, $\vert D' \cap \{a_c^p,c_c^p\} \vert = 1$ for any $p \in \{x,y,z\}$ which implies that $c_c^p \in D'$ for any $p \in \{x,y,z\}$ for otherwise at least one of $d_c^x,d_c^y$ and $d_c^z$ would not be dominated. It follows that $D' \cap \{a_c^x,a_c^y,a_c^z\} = \varnothing$ and so $u_c$ is not dominated, a contradiction. Thus, there exists $\ell \in \{x,y,z\}$ such that $t_\ell^c \in D'$, say $x$ without loss of generality. Then by Observation~\ref{obs:FxTx}(i), $T_x \in D'$ and so necessarily $F_x \notin D'$ by Remark~\ref{rem:Gphi}. But then, $f_x^c \notin D'$ for otherwise by Observation~\ref{obs:FxTx}(ii), $F_x$ would belong to $D'$, and so $g_c^x \in D'$ ($f_c^x$ would otherwise not be dominated). It then follows from Observations~\ref{obs:clause} and \ref{obs:clause2} that $D' \cap \{f_c^x,f_c^y,f_c^z,g_c^x,g_c^y,g_c^z,v_c,w_c\} = \{g_c^x,v_c\}$ which implies that for $p \in \{y,z\}$, $f_p^c \in D'$ for otherwise $f_c^p$ would not be dominated; in particular, $F_p \in D'$ for $p \in \{y,z\}$ by Observation~\ref{obs:FxTx}(ii).~$\diamond$\\

Combining Remark ~\ref{rem:Gphi} and Observations \ref{obs:FxTx} and \ref{obs:sat}, we conclude that for any variable $x \in X$, $\vert D \cap \{T_x,F_x\} \vert = 1$ and for any clause $c$ containing variables $x,y$ and $z$, there exists exactly one variable $\ell \in \{x,y,z\}$ such that $T_\ell \in D$. Therefore, we may construct a truth assignment satisfying $\Phi$ as follows: for any variable $x \in X$, if $T_x \in D$ we set $x$ to true, otherwise we set $x$ to false. This concludes the proof of Claim~\ref{clm:phisat}.
\end{claimproof}

\begin{claim}
\label{clm:GPhi}
$\gamma_t(G_{\Phi}) = 14 \vert X \vert + 8 \vert C \vert$ if and only if $G_{\Phi}$ is a \no-instance for \contractd.
\end{claim}

\begin{claimproof}
Assume first that $\gamma_t(G_{\Phi}) = 14 \vert X \vert + 8 \vert C \vert$ and let $D$ be a minimum total dominating set of $G_{\Phi}$ (note that Remark~\ref{rem:Gphi} holds). Let us show that $D$ contains no $P_3$. 

First, consider a clause $c \in C$ containing variables $x,y$ and $z$. Note that by Observation \ref{obs:clause2} and Remark \ref{rem:Gphi},  $\vert D \cap V(G_c^F) \vert = 2$ and thus $D \cap V(G_c^F)$ cannot contain any $P_3$ nor can it be part of a $P_3$. Now by Observation~\ref{obs:sat}, there exists $\ell \in \{x,y,z\}$ such that $t_\ell^c \in D$ and for $p \in \{x,y,z\} \setminus \{\ell\}$, $f_p^c \in D$. Assume without loss of generality that $\ell = x$ and denote by $c'$ and $c''$ the two other clauses in which $x$ occurs. It follows from Observation \ref{obs:FxTx} and Remark \ref{rem:Gphi} that $T_x\in D$ and $F_x\not\in D$. Then, necessarily $t_x^p \in D$ for $p \in \{c',c''\}$; indeed, since by Observation~\ref{obs:sat}, there exists a variable $t$ contained in $c'$ such that $t_t^{c'} \in D$ and $f_r^{c'} \in D$ for the other variables $r \neq t$ in $c'$, necessarily $t=x$ for otherwise we would conclude by Observation~\ref{obs:FxTx} that $F_x \in D$, a contradiction (the same reasoning applies for $c''$). It then follows from Observation \ref{obs:variable2}(i) that $D \cap (\{t_x^p,p \in \{c,c',c''\}\} \cup \{d_x^p,p \in \{c,c',c''\}\} \cup \{c_x^p,p \in \{c,c',c''\}\} \cup \{b_x^p,p \in \{c,c',c''\}\} \cup \{a_x^p,p \in \{c,c',c''\}\} \cup \{T_x,F_x,u_x,v_x\}) = \{t_x^p,p \in \{c,c',c''\}\} \cup \{d_x^p,p \in \{c,c',c''\}\} \cup \{T_x,u_x\}$. On the other hand, since by Observation~\ref{obs:clause2}, $f_p^x \notin D$ for any $p \in \{c,c',c''\}$, necessarily $j_x^p \in D$ ($f_x^p$ would otherwise not be dominated). But then, $h_x^p \in D$ for any $p \in \{c,c',c''\}$ as $j_x^p$ and $g_x^p$ should be dominated (recall that $F_x \notin D$) and so, $D \cap (\{f_x^p,p \in \{c,c',c''\}\} \cup \{j_x^p,p \in \{c,c',c''\}\} \cup \{i_x^p,p \in \{c,c',c''\}\} \cup \{h_x^p,p \in \{c,c',c''\}\} \cup \{g_x^p,p \in \{c,c',c''\}\}) = \{j_x^p,p \in \{c,c',c''\}\} \cup \{h_x^p,p \in \{c,c',c''\}\}$. Thus, $D \cap V(G_x)$ does not contain any $P_3$. Now denote by $k$ and $k'$ the two others clauses in which $y$ occurs. Then, a reasoning similar to the above shows that $f_y^p \in D$ for $p \in \{k,k'\}$ (recall that by assumption, $f_y^c \in D$) and so by Observations~\ref{obs:variable} and \ref{obs:variable2}(iii), we conclude that $D \cap (\{f_y^p,p \in \{c,k,k'\}\} \cup \{j_y^p,p \in \{c,k,k'\}\} \cup \{i_y^p,p \in \{c,k,k'\}\} \cup \{h_y^p,p \in \{c,k,k'\}\} \cup \{g_y^p,p \in \{c,k,k'\}\} \cup \{T_y,F_y,u_y,v_y\}) = \{f_y^p,p \in \{c,k,k'\}\} \cup \{j_y^p,p \in \{c,k,k'\}\} \cup \{F_y,u_y\}$. On the other hand, since $T_y \notin D$ necessarily $t_y^p \notin D$ for any $p \in \{c,k,k'\}$ (we would otherwise conclude by Observation~\ref{obs:FxTx}(i) that $T_y \in D$). We claim that then $d_y^p \in D$ for all $p \in \{c,k,k'\}$. Indeed, if $d_y^p \notin D$ for some $p \in \{c,k,k'\}$ then necessarily $t_p^y \in D$ as $t_y^p$ should be dominated. But then since $t_y^p \notin D$, it must be that $d_p^y \in D$ for otherwise $t_p^y$ would not be dominated. But $\vert D \cap \{t_p^y,d_p^y,c_p^y,b_p^y,a_p^y\} \vert = 2$ and so $D \cap \{a_p^y,c_p^y\} = \varnothing$ thereby contradicting Observation~\ref{obs:clause}(ii). Thus $d_y^c,d_y^k,d_y^{k'} \in D$ which implies that $b_y^p \in D$ for any $p \in\{c,k,k'\}$ as $d_y^p$ and $a_y^p$ should be dominated (recall that $T_y \notin D$). In particular, $D \cap V(G_y)$ does not contain any $P_3$ (the same reasoning shows that $D \cap V(G_z)$ does not contain any $P_3$ either). Now since $t_y^c \notin D$ necessarily $d_c^y \in D$ as $t_c^y$ should be dominated. We then conclude by Remark \ref{rem:Gphi} that $c_c^y \in D$; indeed, $\vert D \cap \{c_c^y,a_c^y\} \vert = 1$ and if $a_c^y \in D$ then $d_c^y$ is not dominated. Similarly, we conclude that $D \cap \{t_c^z,d_c^z,c_c^z,b_c^z,a_c^z\} = \{d_c^z,c_c^z\}$. Thus, since $u_c$ should be dominated as well as $a_c^x$ and $d_c^x$, we obtain by Observation~\ref{obs:clause2} that $D \cap \{u_c,a_c^x,b_c^x,c_c^x,d_c^x,t_c^x\} = \{a_c^x,c_c^x\}$. Thus, $D \cap V(G_c^T)$ contains no $P_3$ nor can it be part of a $P_3$ and so, $D$ contains no $P_3$.\\

Conversely, assume that $G_{\Phi}$ is a \no-instance for \contractd\ and let $D$ be a minimum total dominating set of $G_{\Phi}$. Consider a variable $x \in X$ contained in clauses $c,c'$ and $c''$. First observe that since $D \cap \{F_x,t_x,u_x,v_x\}$ does not contain a $P_3$, it follows from Observation \ref{obs:variable}(i) that $\vert D \cap \{F_x,T_x,u_x,v_x\} \vert = 2$, and we conclude by Observation~\ref{obs:variable}(i) that in fact equality holds (recall that $u_x\in D$). Similarly for any $p \in \{c,c',c''\}$, $\vert D \cap \{a_x^p,b_x^p,c_x^p,d_x^p,t_x^p\} \vert \leq 3$. Now suppose to the contrary that there exists $p \in \{c,c',c''\}$ such that $\vert D \cap \{a_x^p,b_x^p,c_x^p,d_x^p,t_x^p\} \vert = 3$. Then $\vert D \cap \{b_x^p,c_x^p,d_x^p\} \vert \leq 1$; indeed, clearly $\vert D \cap \{b_x^p,c_x^p,d_x^p\} \vert < 3$ and if $\vert D \cap \{b_x^p,c_x^p,d_x^p\} \vert = 2$ then we may assume without loss of generality that $D \cap \{b_x^p,c_x^p,d_x^p\} = \{b_x^p,d_x^p\}$. But then $D \cap \{t_x^p,a_x^p\} \neq \varnothing$ and so, $D \cap \{a_x^p,b_x^p,c_x^p,d_x^p,t_x^p\}$ contains a $P_3$, a contradiction. Thus $\vert D \cap \{b_x^p,c_x^p,d_x^p\} \vert \leq 1$ and we conclude by Observation~\ref{obs:variable}(ii) that $\vert D \cap \{b_x^p,c_x^p,d_x^p\} \vert = 1$ (in fact, either $b_x^p \in D$ or $d_x^p \in D$). It follows that $t_x^p,a_x^p \in D$ and so, necessarily $T_x \notin D$ for otherwise $a_x^p,T_x,u_x$ would induce a $P_3$ (recall that by Observation~\ref{obs:variable}(i), $u_x \in D$). But then, $(D \setminus \{a_x^p,b_x^c\}) \cup \{T_x,d_x^c\}$ is a minimum total dominating set of $G_{\Phi}$ containing a $P_3$, a contradiction. Thus, we conclude that for any $p \in \{c,c',c''\}$, $\vert D \cap \{a_x^p,b_x^p,c_x^p,d_x^p,t_x^p\} \vert \leq 2$; and by symmetry, we also conclude that $\vert D \cap \{g_x^p,h_x^p,i_x^p,j_x^p,f_x^p\} \vert \leq 2$ for any $p \in \{c,c',c''\}$. It then follows from Observation~\ref{obs:variable} that for any variable $x \in X$, $\vert D \cap V(G_x) \vert = 14$. 

Consider now a clause $c \in C$ containing variables $x,y$ and $z$. First observe that since $D \cap \{g_c^x,g_c^y,g_c^z,v_c,w_c\}$ does not contain a $P_3$, it follows from Observation \ref{obs:clause}(i) that $\vert D \cap \{g_c^x,g_c^y,g_c^z,v_c,w_c\} \vert =2$ and $v_c\in D$. Now if there exists $p \in \{x,y,z\}$ such that $f_c^p \in D$ then $g_c^p \notin D$ and so, $f_p^c \in D$ for otherwise $f_c^p$ would not be dominated. It follows that $j_p^c \notin D$ ($D$ would otherwise contain a $P_3$) and $i_p^c \notin D$ ($(D \setminus \{i_p^c\} \cup \{j_p^c\})$ would otherwise contain a $P_3$) which implies that $h_p^c \in D$ as $i_p^c$ would otherwise not be dominated. But then, $(D \setminus \{f_c^p\}) \cup \{j_p^c\}$ is a minimum total dominating set of $G_{\Phi}$ containing a $P_3$, a contradiction. Thus, $\vert D \cap V(G_c^F) \vert \leq 2$. Now for any $p \in \{x,y,z\}$, $\vert D \cap \{a_c^p,b_c^p,c_c^p,d_c^p,t_c^p\} \vert \leq 3$ as $D \cap \{a_c^p,b_c^p,c_c^p,d_c^p,t_c^p\}$ would otherwise contain a $P_3$. Suppose to the contrary that there exists $p \in \{x,y,z\}$ such that $\vert D \cap \{a_c^p,b_c^p,c_c^p,d_c^p,t_c^p\} \vert = 3$. If $\vert D \cap \{a_c^p,b_c^p,c_c^p\} \vert = 1$ then $d_c^p,t_c^p \in D$ and so necessarily $c_c^p \notin D$. Then, since $b_c^p$ should be dominated, it must be that $a_c^p \in D$. But then, $t_p^c \notin D$ ($t_p^c,t_c^p$ and $d_c^p$ would otherwise induce a $P_3$) and so $D' = (D\setminus\{d_c^p\}) \cup \{t_p^c\}$ is a minimum total dominating set of $G_{\Phi}$ with $\vert D' \cap \{t_p^c,d_p^c,c_p^c,b_p^c,a_p^c\} \vert \geq 3$, thereby contradicting the above. Thus, $\vert D \cap \{a_c^p,b_c^p,c_c^p\} \vert = 2$ (indeed, clearly $\vert D \cap \{a_c^p,b_c^p,c_c^p\} \vert < 3$) and we may assume without loss of generality that $D \cap \{a_c^p,b_c^p,c_c^p\} = \{a_c^p,c_c^p\}$. It follows that $d_c^p \notin D$ ($d_c^p,a_c^p$ and $c_c^p$ would otherwise induce a $P_3$) and so, $t_c^p \in D$. But then, it must be that $t_p^c \in D$ ($t_c^p$ would otherwise not be dominated) and so $d_p^c \notin D$ ($d_p^c,t_p^c$ and $t_c^p$ would otherwise induce a $P_3$). It follows that $D' = (D \setminus \{t_c^p\}) \cup \{d_p^c\})$ is a minimum total dominating set of $G_{\Phi}$ with $\vert D' \cap \{t_p^c,d_p^c,c_p^c,b_p^c,a_p^c\} \vert \geq 3$ thereby contradicting the above. Thus for any $p \in \{x,y,z\}$, $\vert D \cap \{a_x^p,b_c^p,c_c^p,d_c^p,t_c^p\} \vert \leq 2$ and we conclude by Observation~\ref{obs:clause}(ii) that in fact equality holds. It follows that $u_c \notin D$; indeed, if $u_c \in D$ then $\vert D \cap \{a_c^x,a_c^y,a_c^z\} \vert \leq 1$ ($D \cap \{u_c,a_c^x,a_c^y,a_c^z\}$ would otherwise contain a $P_3$) which implies that $D' = (D\setminus \{u_c\}) \cup \{a_c^p\}$ with $p \in \{x,y,z\}$ such that $a_c^p \notin D$, is a minimum total dominating set of $G_{\Phi}$ with $\vert D' \cap  \{a_c^p,b_c^p,c_c^p,d_c^p,t_c^p\} \vert \geq 3$ thereby contradicting the above. Thus, we conclude that $\vert D \cap V(G_c) \vert = 8$ for any clause $c \in C$ and so, $\gamma_t(G_{\Phi}) = \vert D \vert = 14 \vert X \vert + 8 \vert C \vert$.  
\end{claimproof}

Now by combining Claims \ref{clm:phisat} and \ref{clm:GPhi}, we obtain that $\Phi$ is satisfiable if and only if $G_{\Phi}$ is a \no-instance for \contractd\, thus concluding the proof.
\end{proof}

\begin{lemma}
\label{lemma:4sub}
Let $G$ be a graph on at least three vertices, and let $G'$ be the graph obtained by 4-subdividing every edge of $G$. Then $ct_{\gamma_t}(G) = 1$ if and only if $ct_{\gamma_t}(G') = 1$.
\end{lemma}

\begin{proof}
Let $G=(V,E)$ be a graph with $\vert V \vert \geq 3$. In the following, given an edge $e=uv$ of $G$, we denote by $e_1$, $e_2$ $e_3$ and $e_4$ the four new vertices resulting from the 4-subdivision of the edge $e$ (where $e_1$ is adjacent to $u$ and $e_4$ is adjacent to $v$). We first prove the following.

\begin{claim}
\label{clm:4sub}
If $H$ is the graph obtained from $G$ by 4-subdividing one edge, then $\gamma_t(H) = \gamma_t(G) + 2$.
\end{claim}

\begin{claimproof}
Assume that $H$ is obtained by 4-subdividing the edge $e =uv$ and consider a minimum total dominating set $D$ of $G$. We construct a total dominating set of $H$ as follows (see Fig. \ref{fig:totdomH}). If $D \cap \{u,v\} = \varnothing$, then $D \cup \{e_2,e_3\}$ is a total dominating set of $H$. If $\vert D \cap \{u,v\} \vert = 1$, say $u \in D$ without loss of generality, then $D \cup \{e_3,e_4\}$ is a total dominating set of $H$. Finally, if $\{u,v\} \subset D$ then $D \cup \{e_1,e_4\}$ is a total dominating set of $H$. We thus conclude that $\gamma_t (H) \leq \gamma_t (G) + 2$.

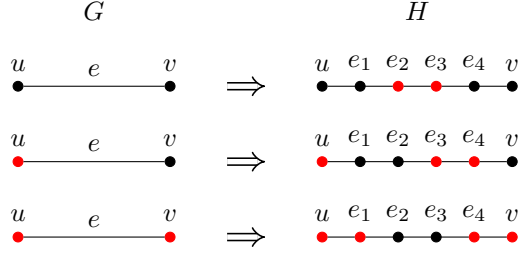
\begin{figure}[htb]
\centering
\begin{tikzpicture}[node distance=.5cm]
\node[circ,label=above:{\small $u$}] (u1) at (0,3) {};
\node[circ,label=above:{\small $v$}] (v1) at (2,3) {};
\draw[-] (u1) -- (v1) node[midway,above] {\small $e$};

\draw[-Implies,line width=.6pt,double distance=2pt] (2.75,3) -- (3.25,3); 

\node[circ,label=above:{\small $u$}] (ub1) at (4,3) {};
\node[circ,label=above:{\small $e_1$},right of=ub1] (e11) {};
\node[circ,red,label=above:{\small $e_2$},right of=e11] (e21) {};
\node[circ,red,label=above:{\small $e_3$},right of=e21] (e31) {};
\node[circ,label=above:{\small $e_4$},right of=e31] (e41) {};
\node[circ,label=above:{\small $v$},right of=e41] (vb1) {};
\draw[-] (ub1) -- (e21)
(e21) -- (e31)
(e31) -- (vb1);

\node[circ,red,label=above:{\small $u$}] (u2) at (0,2) {};
\node[circ,label=above:{\small $v$}] (v2) at (2,2) {};
\draw[-] (u2) -- (v2) node[midway,above] {\small $e$};

\draw[-Implies,line width=.6pt,double distance=2pt] (2.75,2) -- (3.25,2); 

\node[circ,red,label=above:{\small $u$}] (ub2) at (4,2) {};
\node[circ,label=above:{\small $e_1$},right of=ub2] (e12) {};
\node[circ,label=above:{\small $e_2$},right of=e12] (e22) {};
\node[circ,red,label=above:{\small $e_3$},right of=e22] (e32) {};
\node[circ,red,label=above:{\small $e_4$},right of=e32] (e42) {};
\node[circ,label=above:{\small $v$},right of=e42] (vb2) {};
\draw[-] (ub2) -- (e32)
(e32) -- (e42)
(e42) -- (vb2);

\node[circ,red,label=above:{\small $u$}] (u3) at (0,1) {};
\node[circ,red,label=above:{\small $v$}] (v3) at (2,1) {};
\draw[-] (u3) -- (v3) node[midway,above] {\small $e$};

\draw[-Implies,line width=.6pt,double distance=2pt] (2.75,1) -- (3.25,1); 

\node[circ,red,label=above:{\small $u$}] (ub3) at (4,1) {};
\node[circ,red,label=above:{\small $e_1$},right of=ub3] (e13) {};
\node[circ,label=above:{\small $e_2$},right of=e13] (e23) {};
\node[circ,label=above:{\small $e_3$},right of=e23] (e33) {};
\node[circ,red,label=above:{\small $e_4$},right of=e33] (e43) {};
\node[circ,red,label=above:{\small $v$},right of=e43] (vb3) {};
\draw[-] (ub3) -- (e13)
(e13) -- (e43)
(e43) -- (vb3);

\node[draw=none] at (1,4) {\small $G$};
\node[draw=none] at (5.25,4) {\small $H$};
\end{tikzpicture}
\caption{Constructing a total dominating set of $H$ from a total dominating set of $G$ (vertices in red belong to the corresponding total dominating set).}
\label{fig:totdomH}
\end{figure}

Conversely, let $D$ be a minimum total dominating set of $H$. First note that if $e_1\in D$ and $u \not\in D$, necessarily $e_2 \in D$ for otherwise $e_1$ would not be dominated. Similarly, if $e_4 \in D$ and $v \not\in D$ then $e_3 \in D$. Thus, if $e_1,e_4\in D$ then $(D \setminus \{e_1,e_2,e_3,e_4\}) \cup \{u,v\}$ is a total dominating set of $G$ of size at most $\gamma_t (H) - 2$. Now suppose that $e_4 \not\in D$. Then, necessarily $e_2 \in D$ for otherwise $e_3$ would not be dominated, and if $v \notin D$ then $e_3 \in D$ for otherwise $e_4$ would not be dominated. Thus, if $e_1 \in D$ and $e_4 \not\in D$, either $v \in D$ in which case $D \setminus \{e_1,e_2,e_3,e_4\}$ is a total dominating set of $G$ of size at most $\gamma_t(H) - 2$; or $v\not\in D$ and $(D \setminus \{e_1,e_2,e_3,e_4\}) \cup \{v\}$ is a total dominating set of size at most $\gamma(H) - 2$. By symmetry, we conclude similarly if $e_4 \in D$ and $e_1 \not\in D$. Now if both $e_1$ and $e_4$ do not belong to $D$ then $e_2, e_3 \in D$ and so, $D \setminus \{e_1,e_2,e_3,e_4\}$ is a total dominating set of $G$ of size at most $\gamma_t (H) - 2$. Therefore, $\gamma_t(G) \leq \gamma_t(H) - 2$ which concludes the proof of the claim.
\end{claimproof}

\begin{remark}
Note that the minimum total dominating set $D$ of $G$ constructed from a minimum total dominating $D'$ of $H$ according to the proof of Claim \ref{clm:4sub} has the following property: if $e_1 \in D'$ (resp. $e_4 \in D'$) then $v \in D$ (resp. $u \in D$).
\end{remark}

We now prove the statement of the lemma. Let $G'$ be the graph obtained by 4-subdividing every edge of $G$. Then, $\gamma_t(G') = \gamma_t(G) + 2 \vert E \vert$ by Claim \ref{clm:4sub}.

First assume that $ct_{\gamma_t}(G) = 1$. Then by Theorem \ref{theorem:1totcontracdom}, there exists a minimum total dominating set $D$ of $G$ containing a $P_3$, say $u,v,w$. Let $D'$ be the minimum total dominating set of $G'$ constructed from $D$ according to the proof of Claim \ref{clm:4sub}. Then $D'$ contains a $P_3$, namely $e_4,v,f_1$ where $e=uv$, $f=vw$ and $f_1$ is the vertex resulting from the 4-subdivision of $f$ adjacent to~$v$.

Conversely, assume that $ct_{\gamma_t}(G') = 1$. Then by Theorem \ref{theorem:1totcontracdom}, there exists a minimum total dominating set $D'$ of $G'$ containing a $P_3$ which we denote by $P$ in the following. Now let $D$ be the minimum total dominating set of $G$ constructed from $D'$ according to the proof of Claim \ref{clm:4sub}. If $P$ is made up of the vertices $e_4,v,f_1$, where $e = uv$, $f =vw$ and $f_1$ is the vertex resulting from the 4-subdivision of $f$ adjacent to $v$, then $u,v,w \in D$ by construction (see Remark 2). If $P$ is made up of the vertices $u, e_1,e_2$, where $e=uv$, then we may assume that $u$ has no other neighbor in $D$ than $e_1$ (we would otherwise fall back into the previous case). Suppose first that $v \in D'$. Then, $e_4 \not\in D'$ for otherwise $D' \setminus \{e_2\}$ would be a total dominating set of $G'$ of size strictly less than that of $D'$, a contradiction. It follows that $v$ has a neighbor $f_1$ belonging to $D$, with $f=vw$ ($v$ would otherwise not be dominated); but then $w \in D$ by construction (see Remark 2) and so, $D$ contains $u,v,w$. Thus, suppose that $v \not\in D'$. Then $e_4 \not\in D'$; indeed, if $e_4 \in D'$ then $e_3 \in D'$ ($e_4$ would otherwise not be dominated) but then, $D' \setminus \{e_2\}$ is a total dominating set of $G'$ of size strictly less than that of $D'$, a contradiction. It follows that $v$ has a neighbor $f_1$ belonging to $D'$, with $f= vw$ ($v$ would otherwise not be dominated). But then, $v,w \in D$ by construction (see Remark 2) and so, $D$ contains $u,v,w$. Suppose finally that $P$ is made up of the vertices $e_1,e_2,e_3$ with $e=uv$ and assume that $u \not\in D'$ (we would otherwise fall back into the previous case). Then $v\not\in D'$ for otherwise $D' \setminus \{e_3\}$ would be a total dominating set of $G'$ of size stricly less than that of $D'$, a contradiction. Suppose first that $e_4 \in D'$. If $v$ has another neighbor in $D'$, say $f_1 \in D'$ with $f= vw$, then by construction $D$ contains $u,v,w$ (see Remark 2). Thus, we may assume that $v$ has no other neighbor in $D'$ than $e_4$. Now since $\vert V \vert \geq 3$ and $G$ is connected, one of $u$ and $v$ has a neighbor in $V \setminus \{u,v\}$, say $f =vw \in E$ without loss of generality. Note that we may assume that $w \not\in D'$ for otherwise $D$ would contain $u,v,w$. Now since $f_1 \not\in D'$ by assumption, necessarily $f_2 \in D'$ ($f_1$ would otherwise not be dominated) and $f_3 \in D'$ ($f_2$ would otherwise not be dominated) and so, by considering $D'' = (D' \setminus \{e_3,e_4\}) \cup \{v,f_1\}$, we fall back into the previous case (indeed, $D''$ contains $v,f_1,f_2$). Second, suppose that $e_4 \not\in D$. 
Clearly, $v$ has a neighbor $f_1 \in D$, with $f=vw$ ($v$ would otherwise not be dominated), and $f_2 \in D$ ($f_1$ would otherwise not be dominated). But then, by considering $D'' = (D' \setminus \{e_3\}) \cup \{v\}$, we fall back into the previous case (indeed, $D''$ contains $v,f_1,f_2$). Thus, $G$ has a minimum total dominating set containing a $P_3$ and we conclude by Theorem \ref{theorem:1totcontracdom} that $ct_{\gamma_t}(G) = 1$. 
\end{proof}

By applying a 4-subdivision to an instance of \contractd{} sufficiently many times, we deduce the following from Lemma \ref{lemma:4sub}.

\begin{theorem}
\label{thm:cycles}
For any $l \geq 3$, \contractd{} is $\mathsf{NP}$-hard when restricted to $\{C_3,\ldots,C_l\}$-free graphs.
\end{theorem}


\section{Algorithms}
\label{sec-easy}

In this section, we deal with the cases in which \contractd{} is tractable. A first simple approach to this problem, from which we obtain Proposition~\ref{prop:boundedtdom}, is based on brute force.

\begin{proposition}
\label{prop:boundedtdom}
\contractd{} can be solved in polynomial-time solvable on a graph class $\mathcal{C}$, if one of the following holds:
\begin{itemize}
\item[(a)] $\mathcal{C}$ is closed under edge contraction and \tdom\ is solvable in polynomial time on $\mathcal{C}$; or
\item[(b)] for every $G \in \mathcal{C}$, $\gamma_t(G) \leq q$ where $q$ is a fixed constant; or
\item[(c)] $\mathcal{C}$ is the class of $(H+K_1)$-free graphs where $\vert V(H)\vert=q$ is a fixed constant and \contractd{} is polynomial-time solvable on $H$-free graphs.
\end{itemize}
\end{proposition}

\begin{proof}
In order to prove (a), it suffices to note that if we can compute $\gamma_t(G)$ and $\gamma_t(G/ e)$ for any edge $e$ of $G$ in polynomial time, then we can determine in polynomial time whether $G$ is a \yes-instance for \contractd.

For (b), we proceed as follows. Given a graph $G$ of $\mathcal{C}$, we first check whether $G$ has a dominating edge. If it is the case, then $G$ is a \no-instance for \contractd{}. Otherwise, we may consider any subset $S \subseteq V(G)$ with $\vert S \vert \leq q$ and check whether it is a total dominating set of $G$. Since there are at most $\Oh(n^q)$ possible such subsets, we can determine the total domination number of $G$ and check whether the conditions given in Theorem~\ref{theorem:1totcontracdom} are satisfied in polynomial time.

So as to prove (c), we provide the following algorithm. Let $H$ and $q$ and let $G$ be an instance of \contractd{} on $(H+K_1)$-free graphs. We first test whether $G$ is $H$-free (note that this can be done in time $\Oh(n^q)$). If this is the case, we use the polynomial-time algorithm for \contractd{} on $H$-free graphs. Otherwise, there is a set $S\subseteq V(G)$ such that $G[S]$ is isomorphic to $H$; but since $G$ is a $(H+K_1)$-free graph, $S$ must then be a dominating set of $G$ and so, $\gamma_t(G)\leq 2q$. We then conclude by Proposition~\ref{prop:boundedtdom}(b) that \contractd{} is also polynomial-time solvable in this case.
\end{proof}

\begin{theorem}
\label{theorem:p5free}
\contractd{} is polynomial-time solvable on $P_5$-free graphs.
\end{theorem}

\begin{proof}
Let $G$ be a $P_5$-free graph. If $\gamma_t(G) = 2$, then $G$ is clearly a \no-instance for \contractd{}. Now, assume that $\gamma_t(G) \geq 3$ and consider a minimum total dominating set $D$ of $G$. Let us now show that then $G$ is a \yes-instance for \contractd{}. If $\gamma_t(G) = 3$ then it is clear that $D$ contains a $P_3$ as every vertex in $S$ has a neighbor in $D$; thus, by Theorem~\ref{theorem:1totcontracdom}, $G$ is a \yes-instance for \contractd{}. Next, suppose that $\gamma_t (G) \geq 4$ and consider $u,v \in D$ such that $d_G(u,v) = \max_{x,y \in D} d_G(x,y)$. If $d_G(u,v) =1$, then $G[D]$ is a clique and $G$ is therefore a \yes-instance for \contractd{} by Theorem \ref{theorem:1totcontracdom}. Thus, we may assume that $d_G(u,v) \geq 2$. Furthermore, we may assume that $u$ and $v$ have no common neighbor in $D$ for otherwise we are done by Theorem \ref{theorem:1totcontracdom}. Denote by $x$ (resp. $y$) a neighbor of $u$ (resp. $v$) in $D$ and assume that $x$ and $y$ are not adjacent (if $xy \in E(G)$ then we are done by Theorem~\ref{theorem:1totcontracdom}). Since $G$ is $P_5$-free, $d_G(x,y) \leq 3$. 

Suppose first that $d_G(x,y) = 3$ and let $a$ (resp. $b$) be the neighbor of $x$ (resp. $y$) on a shortest path from $x$ to $y$. Then $u$ is adjacent to either $a$ or $b$ but not both; indeed, $u$ is adjacent to either $a$ or $b$ as $u,x,a,b,y$ would otherwise induce a $P_5$. By symmetry, the same holds for $v$. But if $u$ is adjacent to both, since $v$ is adjacent to either $a$ or $b$, we would have $d_G(u,v) = 2 < d_G(x,y)$ thereby contradicting the choice of $u$ and $v$. A similar reasoning shows that if $u$ is adjacent to $a$ (resp. $b$) then $v$ is adjacent to $b$ (resp. $a$). Assume without loss of generality that $u$ is adjacent to $a$ (and thus, $v$ is adjacent to $b$). Then, $N_G(u) \cup N_G(v) \subseteq N_G(a) \cup N_G(b)$; indeed, if $t$ is a neighbor of $u$ then $t$ is nonadjacent to $v$ (recall that $d_G(u,v) \geq d_G(x,y) = 3$) and thus, $t$ is adjacent to either $a$ or $b$ for otherwise $t,u,a,b,v$ would induce a $P_5$. We conclude similarly if $t$ is a neighbor of $v$. But then, $D'=(D \backslash \{u,v\}) \cup \{a,b\}$ is a minimum total dominating set of $G$ containing a $P_3$; indeed, $D'$ is clearly dominating and if a vertex $w \in D$ were dominated by either $u$ or $v$, then $w$ is dominated $a$ or $b$ in $D'$. We then conclude by Theorem \ref{theorem:1totcontracdom} that $G$ is a \yes-instance for \contractd{}. 

Now, suppose that $d_G(x,y) = 2$ and denote by $a$ the vertex on a shortest path from $x$ to $y$. Then, $a$ is adjacent to either $u$ or $v$ for otherwise $u,x,a,y,v$ induce a $P_5$. Suppose first that $a$ is adjacent to both $u$ and $v$. We may assume that both $x$ and $y$ have at least one private neighbor with respect to $D$; if it weren't the case for $x$, then $(D \backslash \{x\}) \cup \{a\}$ would be a minimum total dominating set of $G$ containing a $P_3$ (the same argument holds for $y$). Let $t$ (resp. $s$) be a private neighbor of $x$ (resp. $y$). Clearly, $t$ and $s$ must be nonadjacent since otherwise, $x,t,s,y,v$ induce a $P_5$. Also, at least $t$ or $s$ is adjacent to $a$, otherwise $t,x,a,y,s$ induce a $P_5$. Without loss of generality, we may assume that $s$ is adjacent to $a$. If $t$ is nonadjacent to $a$ then every private neighbor $r$ of $y$ must be adjacent to $a$ for otherwise $t,x,a,y,r$ would induce a $P_5$; thus, $(D \backslash \{y\}) \cup \{a\}$ is a minimum total dominating set of $G$ containing a $P_3$ and so by Theorem~\ref{theorem:1totcontracdom}, we have that $ct_{\gamma_t} (G) = 1$. Thus, we may assume now that $t$ is also adjacent to $a$, and hence every private neighbor of $x$ and $y$ is adjacent to $a$ (if there exists a private neighbor of $x$ or $y$ which is nonadjacent to $a$, we conclude as previously), and therefore $(D \backslash \{y\}) \cup \{a\}$ is a minimum total dominating set of $G$ containing a $P_3$. The result then follows from Theorem~\ref{theorem:1totcontracdom}. Now, if $a$ is nonadjacent to one of $u$ and $v$, say $v$ without loss of generality, then any neighbor $t$ of $u$ is adjacent to either $a$, $y$ or $v$ for otherwise $tuayv$ would induce a $P_5$; but then, $(D\backslash \{u\}) \cup \{a\}$ is a minimum total dominating set of $G$ containing a $P_3$ and thus by Theorem~\ref{theorem:1totcontracdom}, $ct_{\gamma_t}(G) = 1$.
\end{proof}

\begin{theorem}
\label{thm:P4kP3}
For any fixed $k \geq 0$, \contractd\ is polynomial-time solvable on $(P_4+kP_3)$-free graphs.
\end{theorem}

\begin{proof}
First observe that since \tdom\ is polynomial-time solvable on $P_4$-free graphs~\cite[Theorem 25]{semidom}, \contractd\ is polynomial-time solvable on $P_4$-free graphs by Proposition~\ref{prop:boundedtdom}(a) (note indeed that the class of $P_4$-free graphs is closed under edge contraction). Now assume that $k \geq 1$ and let $G$ be a $(P_4+kP_3)$-free graph containing an induced $P_4+(k-1)P_3$. Let $A\subseteq V(G)$ be such that $G[A]$ is isomorphic to $P_4+ (k-1)P_3$, set $B$ to be the set of vertices at distance one from $A$ and let $C = V(G) \setminus (A \cup B)$. Note that since $G$ is $(P_4+kP_3)$-free, $G[C]$ is a disjoint union of cliques. 

Let $\mathcal{K}$ be the set of maximal cliques in $G[C]$. Let $\mathcal{K}'\subseteq \mathcal{K}$ be the subset of cliques such that the closed neighborhood of each clique in $\mathcal{K}'$ does not contain an induced $P_4$ and no clique in $\mathcal{K}'$ is complete to a vertex in $B$. We call a clique $K\in \mathcal{K'}$ a \emph{regular clique} if there exist $k$ other cliques $K_1,\ldots,K_{k}\in\mathcal{K'}$ such that $K,K_1,\ldots,K_{k}$ have pairwise distance at least four from one another. We denote by $\mathcal{R}$ the set of regular cliques. Note that we can identify this set in polynomial time.

\begin{claim}\label{HalfChoosable}
Let $K\in \mathcal{K}$ be a clique of at least two vertices such that $N[K]$ is $P_4$-free. Then there exist $x,y\in K$ such that $N[xy]=N[K]$ and for any $v\in N[K]$, either $N[K]\subseteq N[vx]$ or $N[K]\subseteq N[vy]$.
\end{claim}

\begin{claimproof}
Let $x,y\in K$ be such that $|N[xy]|$ is maximum amongst all pairs of vertices in $K$. Suppose for a contradiction that there exists a vertex $b\in N[K]$ nonadjacent to both $x$ and $y$, and let $c\in K$ be a neighbor of $b$. Suppose  that $x$ has a neighbor $p_x$ which is adjacent to neither $c$ nor $y$, and that $y$ has a neighbor $p_y$ which is adjacent to neither $c$ nor $x$. Then $p_x$ and $p_y$ must be adjacent for otherwise $p_x,x,y,p_y$ would induce a $P_4$; but then, $p_x,p_y,y,c$ induce a $P_4$, a contradiction. It follows that $N[xy]$ is dominated by either $cx$ or $cy$, say $N[xy] \subseteq N[cx]$ without loss of generality. Now since $b \in N[cx] \setminus N[xy]$, we conclude that $|N[xy]| < |N[cx]|$ thereby contradicting the maximality of $|N[xy]|$. Hence, $xy$ dominates $N[K]$ and thus, $N[xy]=N[K]$.
Now consider $v\in N[K]\setminus\set{x,y}$ and assume without loss of generality that $v$ is adjacent to $x$. Suppose that $N[xy]\not\subseteq N[vx]$, that is, $y$ has a neighbor $p_y$ which is adjacent to neither $v$ nor $x$. Then $y$ must be adjacent to $v$ for otherwise $v,x,y,p_y$ would induce a $P_4$. Now if there exists a vertex $p_x\in N(x)$ which is adjacent to neither $y$ nor $v$, then either $p_x$ is not adjacent to $p_y$ in which case $p_x,x,y,p_y$ induce a $P_4$, or $p_x$ is adjacent to $p_y$ and $v,x,p_x,p_y$ induce a $P_4$, a contradiction in both cases. Thus, we conclude that if $N[xy]\not\subseteq N[vx]$ then $N[xy]\subseteq N[vy]$.
\end{claimproof}

\begin{claim}\label{V'Adjacency}
Let $K_1,\ldots, K_{k+1} \in\mathcal{R}$ be regular cliques which are pairwise at distance at least four from one another. For any $i\in[k+1]$, if $v\in V(G)\setminus N[K_i]$ is adjacent to a vertex in $N(K_i)\cap B$ then there exists $j\in[k+1]$, $j\neq i$, such that $v$ is complete to $N(K_j)\cap B$.
\end{claim}

\begin{claimproof}
Assume that there exists a vertex $v\in V(G) \setminus N[K_i]$ where $i\in[k+1]$, which is adjacent to a vertex $b_i\in N(K_i)\cap B$. Let $c_i\in K_i\cap N(b_i)$ and $c'_i\in K_i\setminus N(b_i)$ (recall that since $K_i$ is a regular clique, $b_i$ is not complete to $K_i$). Suppose for a contradiction that there exists no $j\in[k+1]$ such that $v$ is complete to $N(K_j)\cap B$. Then for every $j\in[k+1] \setminus \{i\}$, there exists a vertex $b_j\in N(K_j)\cap B$ which is nonadjacent to $v$. For every $j\in [k+1] \setminus \{i\}$, let $c_j\in K_j\cap N(b_j)$ and $c'_j\in K_j\setminus N(b_j)$. Then $\bigcup_{i=1}
^{k+1}\{b_i,c_i,c'_i\}\cup \{v\} $ induces a $P_4+kP_3$, a contradiction.
\end{claimproof}

\begin{claim}\label{RegularCliquesHaveTwoVerticesFromD}
Let $D$ be a minimum total dominating set of $G$ and let $K\in\mathcal{R}$ be a regular clique. Then $|D\cap N[K]|=2$.
\end{claim}

\begin{claimproof}
It is clear from the definition that $\vert D \cap N[K] \vert \geq 2$ for any regular clique $K$. Now suppose for a contradiction that there exists a clique $K_1 \in \mathcal{R}$ such that $\vert D \cap N[K_1] \vert \geq 3$ and let $K_2,\ldots,K_{k+1}$ be $k$ regular cliques such that $K_1,K_2,\ldots,K_{k+1}$ are pairwise at distance at least four from one another. It follows from the above that \[\left\vert D \cap \bigcup_{2 \leq i \leq k+1} N[K_i] \right\vert \geq 2k.\] Now by \Cref{HalfChoosable}, we have that for any $i\in [k+1]$, there exist $b_i \in B \cap N[K_i]$ and $c_i \in N(b_i) \cap K_i$ such that $b_ic_i$ dominates $N[K_i]$. But then $D' = (D \setminus \bigcup_{i \in [k+1]} N[K_i]) \cup \bigcup_{i \in [k+1]} \{b_i,c_i\}$ is a total dominating set of G; indeed, if $v \in V(G)$ is adjacent to a vertex in $D \cap N[K_i]$, for some $i \in [k+1]$, then either $v \in N[K_i]$ in which case $v \in N[b_ic_i]$, or $v \in V(G) \setminus N[K_i]$ and we conclude by Claim~\ref{V'Adjacency} that $v$ is complete to $N(K_j) \cap B$ for some $j \in [k+1]$ (in particular, $v$ is adjacent to $b_j$). But $\vert D' \vert < \vert D \vert$, a contradiction to the minimality of $D$.
\end{claimproof}

\begin{remark}
\label{rem:tdset}
Note that by the proof of Claim \ref{RegularCliquesHaveTwoVerticesFromD}, we have that for any minimum total dominating set $D$ of $G$, any $k+1$ regular cliques $K_1,\ldots,K_{k+1}\in\mathcal{R}$ which are pairwise at distance at least four from one another, any $b_i\in B\cap N(K_i)$ and $c_i\in V(K_i) \cap N(b_i)$ such that $N[K_i]\subseteq N[b_ic_i]$ with $i\in[k+1]$, 
\[\bigcup_{i=1}^{k+1}\set{b_i,c_i}\,\cup \left( D\setminus \bigcup_{1\leq i\leq k+1} N[K_i]\right)\]
is a minimum total dominating set of $G$.
\end{remark}

\begin{claim}\label{AllRegularCliquesAtDistanceFour}
If there are two regular cliques at distance at most three from one another then $G$ is a \yes-instance for \contractd.
\end{claim}

\begin{claimproof}
Assume that such two regular cliques exist and let $K_1$ and $K'_1$ be two regular cliques such that $d(K_1,K'_1) = \min_{K,K'\in\mathcal{R}} d(K,K')$ (note that by assumption, $d(K_1,K'_1) \leq 3$). Now suppose to the contrary that $G$ is a \no-instance for \contractd. Let $K_2,\ldots, K_{k+1}$ (resp. $K'_2,\ldots,K'_{k+1}$) be $k$ regular cliques such that $K_1,\ldots,K_{k+1}$ (resp. $K'_1,\ldots,K'_{k+1}$) have pairwise distance at least four from one another, and denote by $\mathcal{S}=K_1\cup\ldots\cup K_{k+1}$ and $\mathcal{S}'=K'_1\cup\ldots\cup K'_{k+1}$. By \Cref{HalfChoosable}, we have that for every $i\in[k+1]$, there exist $b_i\in N(K_i)\cap B$ and $c_i\in K_i \cap N(b_i)$ (resp. $b'_i\in N(K'_i)\cap B$ and $c'_i\in K'_i \cap N(b'_i)$) such that $N[K_i]\subseteq N[b_ic_i]$ (resp. $N[K'_i]\subseteq N[b'_ic'_i]$). In the following, let $D$ be a minimum total dominating set of $G$.

Suppose first that $K_1$ and $K'_1$ have a common neighbor $v$, that is, $d(K_1,K'_1) =2$. Then by \Cref{HalfChoosable}, there exists $c\in K_1 \cap N(v)$ such that $N[K_1] \subset N[cv]$. As $c$ has a neighbor in $N(K'_1) \cap B$ (namely $v$), it follows from \Cref{V'Adjacency} that there exists $j\in[k+1]$ such that $c$ is complete to $N[K'_j]\cap B$. By \Cref{rem:tdset}, we then have that $D'= (D\setminus N[\mathcal{S}]) \cup \set{c,v,b_2,c_2,\ldots,b_{k+1},c_{k+1}}$ is a minimum total dominating set of $G$. Similarly by \Cref{rem:tdset}, we conclude that $D''= (D'\setminus N[\mathcal{S}']) \cup\set{b'_1,c'_1,\ldots,b'_{k+1},c'_{k+1}}$ is a minimum total dominating set of $G$. But as $c$ belongs to $D''$ and is adjacent to $b'_j$, it follows that $D''$ contains a $P_3$, a contradiction by \Cref{theorem:1totcontracdom}. Thus, $d(K_1,K'_1) \geq 3$; in particular, no two regular cliques in $G$ have a common neighbor by minimality of $d(K_1,K'_1)$. 

Now since $d(K_1,K'_1) \leq 3$ by assumption, there must exist $w_1 \in N(K_1) \cap B$ and $w'_1 \in N(K'_1) \cap B$ such that $w_1$ and $w'_1$ are adjacent. By \Cref{HalfChoosable}, there exist $v_1 \in K_1 \cap N(w_1)$ and $v'_1 \in K'_1 \cap N(w'_1)$ such that $N[K_1] \subseteq N[w_1v_1]$ and $N[K'_1] \subseteq N[w'_1v'_1]$. But then by \Cref{V'Adjacency}, we have that $D'= (D \setminus N[\mathcal{S}\cup\mathcal{S}']) \cup \{w_1,v_1,w'_1,v'_1\} \cup\bigcup_{i=2}
^{k+1}\set{b_i,b'_i,c_i,c'_i}$ is a total dominating set of $G$ of size at most that of $D$; indeed, since no two regular cliques have a common neighbor, it follows from \Cref{RegularCliquesHaveTwoVerticesFromD} that $\vert D \cap N[\mathcal{S} \cup \mathcal{S}'] \vert = 2\left\vert\set{K_1,\ldots,K_{k+1},K'_1,\ldots,K'_{k+1}}\right\vert$ (Note that this number is not necessarily $2(k+1)$). But $D'$ contains a $P_3$, a contradiction by \Cref{theorem:1totcontracdom}.
\end{claimproof}

In the following, given a total dominating set $D$ of $G$, we call an edge $xy$, with $x,y \in D$, a \textit{$B$-edge} (resp. \textit{$C$-edge}; \textit{$B-C$-edge}) if $x,y\in B$ (resp. $x,y\in C$; $x\in B$ and $y\in C$). Given a $B-C$-edge $xy$, we call its endvertices the $B$-vertex and the $C$-vertex, according to the set in which they are contained. Recall that by \Cref{theorem:1totcontracdom}, if $G$ is a \no-instance for \contractd\ then every minimum total dominating set is an induced matching. 

\begin{claim}\label{BEdgesHavePNInC}
Let $D$ be a minimum total dominating set of $G$. If $G$ is a \no-instance for \contractd\ then there are at most $\left(k+1\right)\left(k+|A|\right)-1$ $B$-edges in $D$ which have no private neighbors in $C$.
\end{claim}

\begin{claimproof}
Assume that $G$ is a \no-instance for \contractd\ and let $\ell=k+\vert A\vert$. Suppose for a contradiction that there exist $\ell\left(k+1\right)$ $B$-edges $x_1y_1,\ldots,x_{\ell(k+1)}y_{\ell(k+1)}$ with no private neighbor in $C$. For every $i \in [k+1]$, let $X_i = \set{x_{(i-1)\ell+1},\ldots,x_{(i-1)\ell+k+\vert A\vert}}$. If for some $i\in [k+1]$ there were no vertex $v\in C$ such that $N(v)\cap D\subseteq X_i$, then $(D\setminus X_i)\cup A$ would be a total dominating set of $G$ containing fewer vertices than $D$, a contradiction. Now consider $v_1 \in C$ such that $N(v_1) \cap D \subseteq X_1$ and let $v_2,\ldots, v_q \in C$ be a longest sequence of vertices defined as follows: for every $i \in [q] \setminus \{1\}$, $N(v_i) \cap D \subseteq X_i$ and $v_i$ is anti-complete to $\{v_1,\ldots,v_{i-1}\}$. We claim that $q=k+1$. Indeed, if $q < k+1$ then every vertex $v \in C$ with $N(v) \cap D \subseteq X_{q+1}$ is adjacent to some vertex in $\{v_1,\ldots,v_q\}$ by maximality of the sequence. But then $(D \setminus X_{q+1}) \cup \{v_1,\ldots,v_q\} \cup A$ is a total dominating set of $G$ of cardinality at most $\vert D \vert$ containing a $P_3$, which cannot be by \Cref{theorem:1totcontracdom}. Thus, $q = k+1$. Now for every $i \in [k+1]$, $v_i$ has to be adjacent to at least two vertices in $X_i$ as no vertex in $X_i$ has a private neighbor in $C$, say $v_i$ is adjacent to $x_{(i-1)\ell + 1}$ and $x_{(i-1)\ell +2}$ without loss of generality. But now the set of vertices 
\[\set{y_1,x_1,v_1,x_2}\cup\bigcup_{i=2}^{k+1}\set{x_{(i-1)\ell+1},x_{(i-1)\ell+2},v_i}\]
induces a $P_4+kP_3$, a contradiction.
\end{claimproof}

Let $D$ be a total dominating set of $G$. We say that a $B$-edge $xy$ can be \emph{turned into a $B-C$-edge} if there exists a vertex $z\in C$ such that $(D\setminus\set{x})\cup\set{z}$ or $(D\setminus\set{y})\cup\set{z}$ is a total dominating set. Analogously, if $xy$ is a $B-C$-edge with $x\in B$ and $y\in C$, we say that $xy$ can be \emph{turned into a $C$-edge} if there is a vertex $z\in C$ such that $(D\setminus\set{x})\cup\set{z}$ is a total dominating set of $G$.

\begin{claim}\label{NotTooManyBEdges}
Let $D$ be a minimum total dominating set of $G$. If $G$ is a \no-instance for \contractd\ and there exists a $B$-edge which has a private neighbor in $C$ and cannot be turned into a $B-C$-edge, then there are at most $|A|+2k$ $B$-edges.
\end{claim}

\begin{claimproof}
Assume that $G$ is a \no-instance for \contractd\ and there exists a $B$-edge $x_1y_1$ such that $x_1$ has a private neighbor $p \in C$ and $x_1y_1$ cannot be turned into a $B-C$-edge. Suppose to the contrary that there exist $|A|+2k$ additional $B$-edges $x_2y_2,\ldots,x_{|A|+2k+1}y_{|A|+2k+1}$. If every private neighbor of $y_1$ is adjacent to $p$, then $(D\setminus\set{y_1})\cup\set{p}$ is a minimum total dominating set of $G$, that is, we can turn the $B$-edge $x_1y_1$ into a $B-C$-edge which is contrary to our assumption. Thus, $y_1$ has a private neighbor $p'$ which is not adjacent to $p$. Set $V_1 = \set{p,x_1,y_1,p'}$ and let $V_2, \ldots, V_q$ be a longest sequence of sets of vertices defined as follows. For every $i \in [q] \setminus \{1\}$, there is a $j_i\in [2,\ldots,\vert A\vert+2k+1]$ and a vertex $v_i\in C$ such that  $V_i=V_{i-1}\cup\set{x_{j_i},y_{j_i},v_i}$, $x_{j_i}\not\in V_{i-1}$, $v_i$ is adjacent to $x_{j_i}$ and $v_i$ is anti-complete to $V_{i-1}\cup{y_{j_i}}$. 
We now claim that $q \geq k + 1$. Indeed, observe first that for any $j\in[2,\ldots,|A|+2k+1]$ such that $x_j \not\in V_q$, every neighbor of $x_j$ in $C$ must be adjacent to some vertex in $V_q \cup \{y_j\}$ by maximality of the sequence. Now if $q < k+1$, there are at least $|A|+k+1$ indices $j\in[|A|+2k+1]$ such that $x_j \notin V_q$ and so, $D' = (D\setminus\set{x_j\colon\, j\in[|A|+2k+1], x_j\not\in V_q})\cup V_q\cup A$ is a total dominating set of $G$ of size at most $\vert D \vert$. But $D'$ contains a $P_3$, a contradiction by \Cref{theorem:1totcontracdom}. Thus, $q\geq k+1$; but now $V_{k+1}$ induces a $P_4+kP_3$, a contradiction.
\end{claimproof}

\begin{claim}\label{NotTooManyBCEdges}
Let $D$ be a minimum total dominating set. If $G$ is a \no-instance for \contractd\ then all but at most $k+|A|$ $B-C$-edges can be turned into $C$-edges.
\end{claim}

\begin{claimproof}
Assume that $G$ is a \no-instance for \contractd. First, suppose for a contradiction that there exist $|A|$ $B-C$-edges $b_1c_1,\ldots,b_{|A|}c_{|A|}$, where $b_i\in B$ and $c_i\in C$ for all $i \in [\vert A \vert]$, such that $c_i$ has no private neighbor in $C$ for every $i\in[|A|]$. Then $(D\setminus\set{c_1,\ldots,c_{|A|}})\cup A$ is a minimum total dominating set of $G$ containing a $P_3$, a contradiction by \Cref{theorem:1totcontracdom}. Thus, there exist at most $|A|-1$ $B-C$-edges such that the $C$-vertex has no private neighbor in $C$. 

Now suppose to the contrary that there exist $k+|A|+1$ $B-C$-edges $b_1c_1,\ldots,b_{k+|A|+1}c_{k+|A|+1}$, where $b_i\in B$ and $c_i\in C$ for all $i \in [k+\vert A \vert +1]$, which cannot be turned into $C$-edges. By the above, we can assume without loss of generality that $c_1$ has a private neighbor $p'$ in $C$. If every private neighbor of $b_1$ is adjacent to $p'$ then $(D \setminus \{b_1\}) \cup \{p'\}$ is a minimum total dominating set of $G$ thereby contradicting the fact that $b_1c_1$ cannot be turned into a $C$-edge. Thus, $b_1$ must have a private neighbor $p$ which is not adjacent to $p'$. We now claim that there are at most $k-1$ indices $i\in\set{2,\ldots,|A|+k+1}$ such that there exists a vertex $v_i\in C$ which is adjacent to $c_i$ but not to $b_i$ and $p$. Indeed, if there were $k$ such indices, say, without loss of generality, for every $i \in [k+1] \setminus \{1\}$, there exists $v_i \in C$ such that $v_i$ is adjacent to $c_i$ but not to $b_i$ and $p$, then $p',c_1,b_1,p,v_2,c_2,b_2,\ldots,v_{k+1},c_{k+1},b_{k+1}$ would induce a $P_4+kP_3$; indeed, since for every $i,j\in[k+1]$, $c_i$ and $c_j$ are non-adjacent and as $v_i$ is contained in $C$ it follows that if $v_i$ is adjacent to $c_i$ then it is non-adjacent to $c_j$ and $v_j$. Thus, we obtain a contradiction. It follows that there are at least $|A|+1$ indices $i\in\set{2,\ldots,|A|+k+1}$ such that every neighbor of $c_i$ in $C$ is adjacent to $b_i$ or $p$, say without loss of generality indices 2
to $|A|+2$. But then $(D\setminus\set{c_2,\ldots,c_{|A|+2}})\cup A\cup\set{p}$ is a minimum total dominating set containing a $P_3$, a contradiction by \Cref{theorem:1totcontracdom}.
\end{claimproof}

\begin{claim}\label{MostEdgesAreCEdges}
If $G$ is a \no-instance for \contractd\ then every minimum total dominating set of $G$ is an induced matching and there exists a minimum total dominating set of $G$ in which all but at most $\left(k+2\right)\left(k+|A|\right)+\vert A\vert-1$ edges are contained in $C$.
\end{claim}

\begin{claimproof}
First note that since $G$ is a \no-instance for \contractd, every minimum total dominating set of $G$ is an induced matching by \Cref{theorem:1totcontracdom}. Now let $D$ be a minimum total dominating set of $G$ containing as few $B$-edges as possible amongst all minimum total dominating sets of $G$. Then either no $B$-edge has a private neighbor in $C$ in which case $D$ contains at most $\left(k+1\right)\left(k+|A|\right)-1$ $B$-edges by \Cref{BEdgesHavePNInC}; or there exists a $B$-edge with a private neighbor in $C$ in which case $D$ contains at most $\vert A \vert + 2k$ $B$-edges by \Cref{NotTooManyBEdges} (indeed, note that by the choice of $D$, no $B$-edge can be turned into a $B-C$-edge). It then follows from \Cref{NotTooManyBCEdges} that we can modify $D$ in order to obtain a minimum total dominating set of $G$ which has at most $k+|A|$ $B-C$-edges, and as the number of edges intersecting $A$ is trivially not more than $\vert A\vert$, the claim follows. 
\end{claimproof}

\begin{claim}\label{P4FreeImpliesCEdgeDominatesneighborhood}
Assume that $G$ is a \no-instance for \contractd{} and let $D$ be a minimum total dominating set of $G$. Let $x_1y_1$ and $x_2y_2$ be two $C$-edges contained in $K_1$ and $K_2$ respectively, where $K_1,K_2\in\mathcal{K}$. If $N[K_1]$ and $N[K_2]$ are $P_4$-free then $K_1$ and $K_2$ have distance at least four from one another.
\end{claim}

\begin{claimproof}
Assume that $N[K_1]$ and $N[K_2]$ are $P_4$-free. First, suppose for a contradiction that $d(K_1,K_2) = 2$, that is, there exists a vertex $b\in N(K_1)\cap N(K_2)$. By \Cref{HalfChoosable}, there exist $c_1\in V(K_1) \cap N(b)$ and $c_2\in V(K_2) \cap N(b)$ such that $N[K_1] \subseteq N[bc_1]$ and $N[K_2] \subseteq N[bc_2]$. But then, $(D\setminus\set{x_1,x_2,y_1,y_2})\cup\set{c_1,c_2,b}$ is a total dominating set of $G$ containing fewer vertices than $D$, a contradiction. Thus, $d(K_1,K_2) > 2$. 

Now suppose for a contradiction that $d(K_1,K_2) = 3$, that is, there exist $b_1\in N(K_1)\cap B$ and $b_2\in N(K_2)\cap B$ such that $b_1$ and $b_2$ are adjacent. By \Cref{HalfChoosable}, there exist $c_1\in V(K_1) \cap N(b_1)$ and $c_2\in V(K_2) \cap N(b_2)$ such that $N[K_1] \subseteq N[c_1b_1]$ and $N[K_2] \subseteq N[c_2b_2]$. But then, $(D\setminus\set{x_1,y_1,x_2,y_2})\cup\set{b_1,b_2,c_1,c_2}$ is a minimum total dominating set containing a $P_3$, a contradiction by \Cref{theorem:1totcontracdom}. Thus, $d(K_1,K_2) > 3$.
\end{claimproof}

\noindent
The following claim is a straightforward corollary of \Cref{P4FreeImpliesCEdgeDominatesneighborhood}.

\begin{claim}\label{P4FreeSuffices}
Assume that $G$ is a \no-instance of \contractd and let $D$ be a minimum total dominating set of $G$. If there exist $k+1$ cliques $K_1,\ldots,K_{k+1} \in \mathcal{K}$ containing $C$-edges such that for any $i \in [k+1]$, $N[K_i]$ is $P_4$-free and there exists no vertex $b \in B$ complete to $K_i$, then $K_1,\ldots,K_{k+1}$ are regular cliques.
\end{claim}

\begin{claim}\label{NotTooManyCEdgesDominatedByASingleVertex}
Assume that $G$ is a \no-instance to \contractd and let $D$ be a minimum total dominating set of $G$. Then there are at most $|A|-1$ cliques $K\in\mathcal{K}$ containing a $C$-edge, for which there exists a vertex $b\in N(K)\cap B$ such that $N[K]\subseteq N[b]$.
\end{claim}
\begin{claimproof}
Suppose for a contradiction that there are $|A|$ cliques $K_1,\ldots,K_{\vert A \vert} \in \mathcal{K}$ containing $C$-edges $x_1y_1,\ldots,x_{\vert A \vert}y_{\vert A \vert}$ respectively, where for any $i \in [\vert A \vert]$, there exists a vertex $b_i\in N(K_i)\cap B$ such that $N[K_i]\subseteq N[b_i]$. Then, $(D\setminus\set{x_1,y_1,\ldots,x_{|A|},y_{|A|}})\cup\set{b_1,\ldots,b_{|A|}}\cup A$ is a minimum total dominating set containing a $P_3$, a contradiction by \Cref{theorem:1totcontracdom}.
\end{claimproof}

\begin{claim}\label{ClaimOfProposition}
Let $D$ be a minimum total dominating set of $G$ and let $\mathcal{C}$ be a set of C-edges pairwise at distance at least three from one another. Then for any subset $E \subseteq \mathcal{C}$ of cardinality $k+1$ for which there exists an induced path $G_E$ of length three containing exactly two vertices from two different edges of $E$, there exist a set $T \subseteq N(\mathcal{C})$ of cardinality at most $2(k-1)$ and a set $N \subseteq\{xy \in \mathcal{C}: N(xy) \cap T =\varnothing\}$ of cardinality at most $\frac{k
^2}{2}-\frac{k}{2}+2$, such that every private neighbor of an edge in $\mathcal{C} \setminus N$ is adjacent to $V(G_E) \cup T$, and $\set{xy\in\mathcal{C}\colon \set{x,y}\cap V(G_{E})\neq\varnothing}\subseteq N$. 
\end{claim}

\begin{claimproof}
Let $E\subseteq \mathcal{C}$ be a set of cardinality $k+1$ for which there exists an induced $P_4$, denoted by $G_E$, containing exactly two vertices from two different edges of $E$, and denote by $xy$ and $x'y'$ the two edges of $E$ such that  $\set{x,y}\cap V(G_E)\neq\varnothing$ and $\set{x',y'}\cap V(G_E)\neq\varnothing$. Let $S_1\subseteq N(\mathcal{C}\setminus \set{xy,x'y'})\cap B$ be a maximum independent set such that every vertex in $S_1$ is adjacent to exactly one endvertex from an edge in $\mathcal{C}$ and not adjacent to $G_E$, and every edge in $\mathcal{C}$ has at most one neighbor in $S_1$. Note that $V(G_E) \cup S_1\cup\set{v,w\colon\,vw\in \mathcal{C},\, N(vw)\cap S_1\neq\varnothing}$ induces a $P_4+|S_1|P_3$ since the vertices in $S_1$ are not adjacent to $G_E$ or to each other by construction, and any edge $vw \in \mathcal{C}$ such that $N(vw) \cap S_1 \neq \varnothing$ has exactly one neighbor $u \in S_1$ where either $u \in N(v) \setminus N(w)$ or $u \in N(w) \setminus N(v)$ by construction. It follows that $|S_1|<k$. We construct a sequence of sets of vertices according to the following procedure.

\begin{itemize}
\item[1.] Initialize $i=1$ and set $C_1=E_1=\set{e\in \mathcal{C}\colon\, N(e)\cap S_1\neq\varnothing}\cup\set{xy,x'y'}$.
\item[2.] Increase $i$ by one.
\item[3.] Let $S_i \subseteq N(\mathcal{C})\cap B \setminus \left(N(G_E)\cup N(C_{i-1})\right)$ be a maximum stable set such that every vertex in $S_i$ is adjacent to exactly one endvertex of an edge in $\mathcal{C}$ and every edge in $\mathcal{C}$ has at most one neighbor in $S_i$. 

Set $E_i=\set{e\in\mathcal{C}\colon\, N(e)\cap S_i\neq\varnothing }$ and $C_i=C_{i-1}\cup E_i$. 
\item[4.] If $|S_i|=|S_{i-1}|$, stop the procedure. Otherwise, return to 2.
\end{itemize}

Consider the value of $i$ at the end of the above procedure (note that $i \geq 2$). Now observe that the following holds: for any edge $vw$ in $\mathcal{C} \setminus C_{i-1}$, every private neighbor $b$ of $vw$ is adjacent to $G_E$ or $S_{i-1}$ for otherwise the procedure would have output $S_{i-1} \cup \{b\}$ in place of $S_{i-1}$. Furthermore, any private neighbor $b$ of an edge in $E_{i-1}$ must be adjacent to $G_E$ or $S_i$ for otherwise the procedure would have output $S_i \cup \{b\}$ in place of $S_{i-1}$ (recall that by construction, $\vert S_{i-1} \vert = \vert S_i \vert$). Thus, it suffices to set $T = S_i \cup S_{i-1}$ and $N = C_{i-2}$ if $i > 2$ (otherwise $N = \set{xy,x'y'}$). Observe that for any $1 \leq p < q \leq i-1$, $\vert S_q \vert < \vert S_p \vert$ which implies that $i \leq k+1$ as $\vert S_1 \vert \leq k-1$; in particular, $\vert S_j \vert \leq k-j$ for any $j \in [i-1]$ and so 
\[ \vert N\vert=\vert C_{i-2} \vert = 2 + \sum_{j=1}^{i-2} \vert E_j \vert = 2 + \sum_{j=1}^{i-2} \vert S_j \vert \leq 2 + \sum_{j=1}^{k-1} (k-j) = \frac{k^2}{2} -\frac{k}{2} + 2.\]
\end{claimproof}

\begin{claim}\label{BigLemma}
Assume that $G$ is a \no-instance of \contractd{} and let $D$ be a minimum total dominating set of $G$. Let $\mathcal{C}$ be a set of $C$-edges pairwise at distance at least three from one another, such that for every subset $E\subseteq\mathcal{C}$ of cardinality $k+1$, there exists an induced $P_4$ containing exactly two vertices from two different edges in $E$. Then \[|\mathcal{C}|< \left(\frac{k
^2}{2} +\frac{3k}{2}\right)\left(\frac{k
^2}{2} +\frac{3k}{2} +1\right)+ k+1.\]
\end{claim}

\begin{claimproof}
Suppose for a contradiction that \[|\mathcal{C}|\geq \left(\frac{k^2}{2} +\frac{3k}{2}\right)\left(\frac{k^2}{2} +\frac{3k}{2} +1\right)+ k+1.\] 
For a set $E\subseteq\mathcal{C}$ of cardinality $k+1$, denote by $G_E$ an induced $P_4$ which contains exactly two vertices from two different edges in $E$. Note that as any two edges in $\mathcal{C}$ have distance at least three from one another, no other edge in $\mathcal{C}$ can be adjacent to $G_E$.

Now let $E_1\subseteq\mathcal{C}$ be a set of $k+1$ edges and let $G_{E_1}$ be as defined above. By \Cref{ClaimOfProposition}, there exist a set $T_1 \subseteq N(\mathcal{C})$ of cardinality at most $2(k-1)$ and a set $N_1 \subseteq\{xy \in \mathcal{C}: N(xy) \cap T_1 =\varnothing\}$ of cardinality at most $\frac{k
^2}{2}-\frac{k}{2}+2$, such that every private neighbor of an edge in $\mathcal{C} \setminus N_1$ is adjacent to $V(G_{E_1}) \cup T_1$, and $\set{vw\in\mathcal{C}\colon \{v,w\}\cap V(G_{E_1})\neq\varnothing}\subseteq N_1$. We construct a sequence of sets of vertices according to the following procedure.

\begin{itemize}
\item[1.] Initialize $i=1$ and set $F_1= N_1\cup\set{e\in\mathcal{C}\colon\, N(e)\cap T_1\neq\varnothing}$.
\item[2.] If $i > \frac{k^2}{2} +\frac{3k}{2}$, stop the procedure. Otherwise, increase $i$ by one.
\item[3.] Let $E_i\subseteq\mathcal{C}\setminus F_{i-1}$ be a set of $k+1$ edges. By \Cref{ClaimOfProposition}, there exist a set $T_i\subseteq B\cap N(\mathcal{C}\setminus F_{i-1})$ of cardinality at most $2(k-1)$ and a set $N_i \subseteq\{xy \in \mathcal{C}: N(xy) \cap T_i =\varnothing\}$ of cardinality at most $\frac{k^2}{2}-\frac{k}{2}+2$, such that every private neighbor of an edge in $\mathcal{C}\setminus\left( F_{i-1}\cup N_i\right)$ is adjacent to $V(G_{E_i}) \cup T_i$, and $\set{vw\in\mathcal{C}\colon \{v,w\} \cap V(G_{E_i})\neq\varnothing}\subseteq N_i$. 

Set $F_i=F_{i-1}\cup N_i\cup\set{e\in\mathcal{C}\colon\, N(e)\cap T_i\neq\varnothing}$ and return to 2.
\end{itemize}

Note that the procedure ends with $i = \frac{k^2}{2} +\frac{3k}{2} + 1$. Indeed, at each iteration, the set $F_j$ increases by at most $\frac{k^2}{2} -\frac{k}{2} + 2+2(k-1)$ and thus at the $j^{th}$ iteration, there are at least 
\[
|\mathcal{C}|-j(\frac{k^2}{2} -\frac{k}{2} + 2+2(k-1)) \geq |\mathcal{C}|-(\frac{k^2}{2} +\frac{3k}{2} + 1)(\frac{k^2}{2} -\frac{k}{2} + 2+2(k-1)) \geq k+1
\]
edges left in $\mathcal{C}$ to form the set $E_{j+1}$. Finally, we increase once more $i$ by one (that is, set $i =  \frac{k^2}{2} +\frac{3k}{2} + 2$) and let $E_i\subseteq\mathcal{C} \setminus F_{i-1}$ be a subset of cardinality $k+1$. Applying \Cref{ClaimOfProposition} with $\mathcal{C}$ (rather than $\mathcal{C} \setminus F_{i-1}$ as above) and $E_i$, we obtain that there exist a set $T_i\subseteq B\cap N(\mathcal{C})$ of cardinality at most $2(k-1)$ and a set $N_i \subseteq\{xy \in \mathcal{C}: N(xy) \cap T_i =\varnothing\}$ of cardinality at most $\frac{k^2}{2}-\frac{k}{2}+2$, such that every private neighbor of an edge in $\mathcal{C}\setminus N_i$ is adjacent to $V(G_{E_i}) \cup T_i$, and $\set{vw\in\mathcal{C}\colon \{v,w\} \cap V(G_{E_i})\neq\varnothing}\subseteq N_i$. Observe that $T_i$ could intersect the sets $V(G_{E_j})$ or $T_j$ for $j<i$, but every private neighbor of an edge in $\mathcal{C}\setminus N_i$ is adjacent to $G_{E_i}$ or $T_i$.

Note that by construction for any $j,j'\in [i-1]$ the sets $V(G_{E_j})\cup T_j$ and $V(G_{E_{j'}})\cup T_{j'}$ are disjoint and there exists no edge in $\mathcal{C}$ which is adjacent to both $V(G_{E_j})\cup T_j$ and $V(G_{E_{j'}})\cup T_{j'}$. Since $T_i$ contains at most $2\left(k-1\right)$ vertices and every vertex in $T_i$ is adjacent to exactly one edge in $\mathcal{C}$, there are at most $2(k-1)$ edges in $\mathcal{C}$ which are adjacent to $T_i$. As for each of these edges there is at most one index $j\in[i-1]$ such that the edge is adjacent to $V(G_{E_j})\cup T_j$, it follows that there are at least $\frac{k^2}{2} -\frac{k}{2} + 2+1$ indices $j \in [i-1]$ such that the sets $V(G_{E_j})\cup T_j$ and $V(G_{E_i})\cup T_i$ are disjoint and there exists no edge in $\mathcal{C}$ which is adjacent to both $V(G_{E_j})\cup T_j$ and $V(G_{E_i})\cup T_i$. Thus, there are at least $\frac{k^2}{2} -\frac{k}{2} + 2+1-\vert N_i\vert\geq 1$ indices $j\in [i-1]$ such that the sets $V(G_{E_j})\cup T_j$ and $V(G_{E_i})\cup T_i$ are disjoint, there exists no edge in $\mathcal{C}$ which is adjacent to both $V(G_{E_j})\cup T_j$ and $V(G_{E_i})\cup T_i$, and every private neighbor of an edge in $\set{vw\in\mathcal{C}\colon\, \set{v,w} \cap N(V(G_{E_j})\cup T_j)\neq\varnothing}$ is adjacent to $V(G_{E_i})\cup T_i$. Let $j$ be one of these indices. By construction, it holds that every private neighbor of an edge in $E_{i}$ is also adjacent to $V(G_{E_j})\cup T_j$. 
Let $x_jy_j$ and $x'_jy'_j$ be the two edges in $\mathcal{C}$ which are adjacent to $G_{E_j}$ and assume, without loss of generality, that $x_j$ and $x'_j$ are contained in $V(G_{E_j})$. Let $x_iy_i$ and $x'_iy'_i$ be the two edges in $\mathcal{C}$ which are adjacent to $G_{E_i}$ and assume, without loss of generality, that $x_i$ and $x'_i$ are contained $V(G_{E_i})$. Let $Y=\set{y\in V(G)\colon\, \exists xy\in \mathcal{C}, x\in N(T_i\cup T_j)}$ be the set of all vertices which are adjacent to neither $T_i$ nor $T_j$ but belong to an edge in $\mathcal{C}$ that is adjacent to $T_i$ or $T_j$. Then, $D\setminus (\{y_i,y'_i,y_j,y'_j\}\cup Y) \cup V(G_{E_i}) \cup V(G_{E_j})\cup T_i\cup T_j$ is a minimum total dominating set containing a $P_3$, a contradiction by \Cref{theorem:1totcontracdom}.
\end{claimproof}

\begin{claim}\label{MostEdgesInRegularCliques}
If $G$ is a \no-instance for \contractd\ then there exists a minimum total dominating set of $G$ in which all but at most 
\[
3\left(|A|-1\right)+\left(\frac{k^2}{2} +\frac{3k}{2}\right)\left(\frac{k^2}{2} +\frac{3k}{2} + 1\right)+ k+1
\] 
edges are contained in regular cliques.
\end{claim}

\begin{claimproof}
Assume that $G$ is a \no-instance for \contractd\ and let $\mathcal{D}$ be the set of all minimum total dominating sets of $G$ in which all but at most $\left(k+2\right)\left(k+|A|\right)+\vert A\vert-1$ edges are $C$-edges (note that $\mathcal{D}$ is nonempty by \Cref{MostEdgesAreCEdges}). For every $D\in\mathcal{D}$ and $S\subseteq B$, let $N_e(D,S)$ be the number of $C$-edges which are adjacent to $S$. For any $D\in\mathcal{D}$, let 
\[
S(D)=\argmax_{S\subseteq B} N_e(D,S) - |S|.
\]
In the following, let $D\in\mathcal{D}$ be a minimum total dominating set of $G$ such that $\max_{S\subseteq B} N_e(D,S)-|S|$ is maximum amongst all sets in $\mathcal{D}$.

First, suppose for a contradiction that there exists a set $S\subseteq B$ which is adjacent to $s=N_e\left(D,S\right)$ $C$-edges $x_1y_1,\ldots,x_sy_s$, such that $s-|S|\geq |A|$. Assume without loss of generality that $S$ is adjacent to $x_i$ for every $i\in[s]$. Then $\left(D\setminus\set{y_1,\ldots,y_s}\right)\cup S\cup A$ is a minimum total dominating set of $G$ containing a $P_3$, a contradiction by \Cref{theorem:1totcontracdom}. 

Now consider a set $S\in S(D)$ of minimum cardinality amongst all sets in $S(D)$. By the above, we have that $N_e\left(D,S\right)\leq |A|-1+|S|$. We claim that every $v\in S$ must be adjacent to at least two $C$-edges to which no other vertex in $S$ is adjacent. Indeed, if $v$ were adjacent to no $C$-edge to which no other vertex in $S$ was adjacent then $S'=S\setminus\set{v}$ would be such that $N_e(D,S')=N_e(D,S)$ and $\vert S'\vert<\vert S\vert$. But then, $N_e(D,S')-\vert S'\vert>N_e(D,S)-\vert S\vert$ thereby contradicting the fact that $S$ belongs to $S(D)$. If $v$ were adjacent to only one $C$-edge to which no other vertex in $S$ is adjacent, then removing $v$ from $S$ would decrease both $|S|$ and $N_e\left(D,S\right)$ by one, leaving the difference unchanged and thus contradicting minimality of $|S|$. It follows that $2|S|\leq N_e\left(D,S\right)$ which combined with the inequality above implies that $|S|<|A|$ and so, $N_e\left(D,S\right)\leq 2\left(|A|-1\right)$. 

Now let $\mathcal{C}$ be the set of $C$-edges which are not adjacent to $S$. We may assume that $D$ satisfies the following property: for any edge $xy \in \mathcal{C}$ contained in a clique $K \in \mathcal{K}$, if there exist $x',y' \in V(K)$ such that $N[xy] \subseteq N[x'y']$ then $N[xy] = N[x'y']$. Indeed, suppose that there exists an edge $xy\in\mathcal{C}$ contained in a clique $K\in\mathcal{K}$ for which there exist $x',y'\in K$ such that $N[xy] \subsetneq N[x'y']$. Then, it suffices to consider $D' = (D \setminus \set{x,y}) \cup \set{x',y'}$ in place of $D$. Clearly, $D'$ is still a minimum total dominating set of $G$ and the value of $\max_{S'\subset B}N_e(D',S')-|S'|$ does not increase by choice of $D$. Furthermore, the value of $\max_{S'\subset B}N_e(D',S')-|S'|$ does not decrease: since $N[xy] \subsetneq N[x'y']$, we have that $N_e(D,S) \leq N_e(D',S)$ and so by choice of $D$ and because $S \in S(D)$, we conclude that $\max_{S'\subseteq B} N_e(D,S')-|S'| = N_e(D,S) - \vert S \vert = N_e(D',S) - \vert S \vert = \max_{S'\subseteq B} N_e(D',S')-|S'|$. In particular, it follows that $x'y'$ cannot be adjacent to $S$ and thus, the set of $C$-edges in $D'$ not adjacent to $S$ has the same cardinality as $\mathcal{C}$. By replacing each such edge of $\mathcal{C}$, the resulting total dominating set satisfies the above property. In the following, we denote by $D$ the resulting dominating set and by $\mathcal{C}$ the set of $C$-edges not adjacent to $S$ for simplicity. 

Now if there were two edges in $\mathcal{C}$ with a common neighbor $v$, then $S'=S\cup\set{v}$ would be such that $\vert S'\vert=\vert S\vert+1$ and $N_e(D,S')=N_e(D,S)+2$. But then, $N_e(D,S')-\vert S'\vert>N_e(D,S)-\vert S\vert$ thereby contradicting the fact that $S$ belongs to $S(D)$. Thus, any two edges in $\mathcal{C}$ are at distance at least three from one another.

Now denote by $\mathcal{C}_2 \subseteq \mathcal{C}$ the set of all $C$-edges $xy$ for which there exists a vertex in $b \in B$ such that $b$ is complete to the clique in $\mathcal{K}$ containing $xy$. Then by \Cref{NotTooManyCEdgesDominatedByASingleVertex,}, $|\mathcal{C}_2|\leq |A|-1$.

Finally, denote by $\mathcal{C}_1\subseteq\mathcal{C}\setminus \mathcal{C}_2$ the set of edges which contain an induced $P_4$ in their closed neighborhood and suppose that $\mathcal{C}_1$ has cardinality at least $k+1$. Now consider a set $E \subseteq \mathcal{C}_1$ of cardinality $k+1$ and denote by $P_e$ an induced $P_4$ contained in $N[e]$, for any $e \in E$. If any two edges in $E$ were pairwise at distance at least four from one another then $\bigcup_{e \in E} P_e$ would be isomorphic to $(k+1) P_4$, a contradiction. Thus, there exist $xy,x'y' \in E$ such that $d(xy,x'y') = 3$, say $d(xy,x'y') =d(x,x')$ without loss of generality; in particular, there exists an induced path $P$ from $x$ to $x'$ of length three containing exactly two vertices (namely $x$ and $x'$) from two different edges of $E$ (namely $xy$ and $x'y'$). Since this holds for any subset $E \subseteq \mathcal{C}_1$ of cardinality $k+1$, we conclude by \Cref{BigLemma} that $|\mathcal{C}_1|< \left(\frac{k
^2}{2} +\frac{3k}{2}\right)\left(\frac{k^2}{2} +\frac{3k}{2} +1\right)+ k+1$.


Now consider an edge $xy \in \mathcal{C}\setminus\left(\mathcal{C}_2\cup\mathcal{C}_1\right)$ which is contained in a clique $K\subseteq\mathcal{K}$ and let us show that $N[xy]=N[K]$. Suppose to the contrary that there is a vertex $b\in N[K]\setminus N[xy]$ and let $c\in K$ be a neighbor of $b$. By assumption on $D$, $N[xy]$ cannot be a strict subset of $N[xc]$ or $N[yc]$, so there exists a vertex $p_x\in N(x)$ which is adjacent to neither $y$ nor $c$ and there exists a vertex $p_y\in N(y)$ which is adjacent to neither $x$ nor $c$. As $xy\notin \mathcal{C}_1$ it follows that $p_x,x,y,p_y$ cannot induce a $P_4$ and thus $p_x$ and $p_y$ must be adjacent. But then $p_x,p_y,y,c$ induce a $P_4$, a contradiction. Hence, every edge $xy\in\mathcal{C}\setminus(\mathcal{C}_1\cup\mathcal{C}_2)$ is contained in a clique in $K \in \mathcal{K}$ whose closed neighborhood is $P_4$-free (recall that $N[xy] = N[K]$ by the above property) and for which there exists no $b \in B$ such that $N[K] \subseteq N[b]$. Furthermore by \Cref{P4FreeSuffices}, either $\mathcal{C}\setminus\left(\mathcal{C}_2\cup\mathcal{C}_1\right)$ has cardinality at most $k$ or contains only regular cliques. Thus, all but at most 
\[
2\left(|A|-1\right)+|\mathcal{C}_2|+\left(\frac{k^2}{2} +\frac{3k}{2}\right)\left(\frac{k^2}{2} +\frac{3k}{2} +1\right)+ k+1.
\]

$C$-edges are contained in regular cliques, which proves the claim since $|\mathcal{C}_2|\leq |A|-1$.
\end{claimproof}

We now present an algorithm which can determine in polynomial-time whether $G$ is a \yes-instance for \contractd\ or not. In the following, we denote by $f(k) = 3\left(|A|-1\right)+\left(\frac{k^2}{2} +\frac{3k}{2}\right)\left(\frac{k^2}{2} +\frac{3k}{2} + 1\right)+ k+1$.

\begin{itemize}
    \item[1.] Determine $A$, $B$, $C$ and $\mathcal{R}$.
     \begin{itemize}
     \item[1.1] If $\mathcal{R}=\varnothing$, then check if there exists a minimum total dominating set of size at most $f(k)$.
     \begin{itemize}
     \item[1.1.1] If the answer is no, then output \yes.
     \item[1.1.2] Else apply Proposition \ref{prop:boundedtdom}(b).
     \end{itemize}
     \item[1.2] Else go to 2.
     \end{itemize}
    \item[2.] Check whether there exist two regular cliques in $\mathcal{R}$ which have distance at most three from one another. If so, output \yes.
    \item[3.] Let $V_1$ be the set of vertices at distance one from $N[\mathcal{R}]$ and let $V_2 = V(G) \setminus (N[\mathcal{R}] \cup V_1)$. If $V_2 = \varnothing$, output \no.
    \item[4.] Determine $\mathcal{S}=\set{S\subseteq V_1 \cup V_2\colon \vert S\vert\leq 2f(k), \forall x\in V_2, N(x)\cap S\neq\varnothing}$. If $\mathcal{S}=\varnothing$, output \yes.
    \item[5.] Let $\mathcal{S}'$ be the family of all sets in $\mathcal{S}$ which have minimum size. 
    \begin{itemize}
    \item[(i)] If there exists a set $S\in\mathcal{S}'$ containing a $P_3$, output \yes. 
    \item[(ii)] If there exists a set $S\in\mathcal{S}'$ such that $S \cap V_1 \neq \varnothing$, output \yes.
    \end{itemize}
    \item[6.] Output \no.
\end{itemize}

Finally, let us show that this algorithm outputs the correct answer. In case $\mathcal{R}=\emptyset$, then \Cref{MostEdgesInRegularCliques} tells us that $G$ is a \yes-instance for \contractd\, if there exists no minimum total dominating set of size at most $f(k)$ (see step 1.1.1). If such a minimum total dominating set exists, then we conclude using Proposition \ref{prop:boundedtdom}(b) (see step 1.1.2). If in step 2, two regular cliques at distance at most three from one another are found then by \Cref{AllRegularCliquesAtDistanceFour}, $G$ is a \yes-instance for \contractd. Otherwise, any two regular cliques have distance at least four to one another and by \Cref{rem:tdset}, there exists a minimum total dominating set $D$ of $G$ such that for any regular clique $K \in \mathcal{R}$, $D \cap N[K] = \{b_K,c_K\}$ where $b _K\in N(K) \cap B$ and $c_K \in K$. In the following, we denote by $D' = \bigcup_{K \in \mathcal{R}} \set{b_K,c_K}$. Note that by \Cref{V'Adjacency}, for any $x \in N[\mathcal{R}] \cup V_1$, $N(x) \cap D' \neq \varnothing$. Now if $V_2 = \varnothing$ then we conclude by \Cref{RegularCliquesHaveTwoVerticesFromD} and the fact that any two regular cliques are at distance at least four from one another, that any minimum total dominating set of $G$ is an induced matching, that is, $G$ is a \no-instance for \contractd\ (see step 3). Otherwise $V_2 \neq \varnothing$ and if $G$ is a \no-instance for \contractd\ then by \Cref{MostEdgesInRegularCliques}, there must exist a set $S \subseteq V_1 \cup V_2$ of cardinality at most $2f(k)$ such that for any $x \in V_2$, $N(x) \cap S \neq \varnothing$. Thus, if $\mathcal{S} = \varnothing$ then $G$ is a \yes-instance for \contractd\ (see step 4). Otherwise $\mathcal{S} \neq \varnothing$, and for any $S \in \mathcal{S}'$, $S \cup D'$ is a minimum total dominating set of $G$. It then follows from \Cref{theorem:1totcontracdom} that if there exists $S \in \mathcal{S}'$ such that $S$ contains a $P_3$ then $G$ is a \yes-instance for \contractd\ (see step 5(i)); otherwise, any $S \in \mathcal{S}'$ is an induced matching and if there exists a set $S \in \mathcal{S}'$ such that $S \cap V_1 \neq \varnothing$ then by \Cref{V'Adjacency}, $S \cup D'$ contains a $P_3$ and so, $G$ is a \yes-instance by \Cref{theorem:1totcontracdom} (see step 5(ii)). Otherwise, for any $S \in \mathcal{S}'$, $S$ is an induced matching and $S \cap V_1 = \varnothing$ which implies that $S \cup D'$ is an induced matching and thus, $G$ is a \no-instance for \contractd. As every step can clearly be done in polynomial time, this concludes the proof.
\end{proof}


\section{Proof of Theorem~~\ref{thm:dichotomyTD}}
\label{sec-proof}

We here give a proof of our main result, Theorem~\ref{thm:dichotomyTD}.\\

Let $H$ be a graph. If $H$ contains a cycle then \contractd\ is $\mathsf{NP}$-hard on $H$-free graphs by Theorem~\ref{thm:cycles}. Thus, we may assume that $H$ is a forest. Now if $H$ contains a vertex of degree at least three, then $H$ contains an induced claw and so, \contractd\ is $\mathsf{coNP}$-hard on $H$-free graphs by Theorem~\ref{thm:tdclaw}. Assume henceforth that $H$ is a linear forest. If $H$ contains a path on at least six vertices, then \contractd\ is $\mathsf{NP}$-hard on $H$-free graphs by Theorem~\ref{thm:P6P5P2}. Thus, we may assume that every connected component of $H$ induces a path on at most five vertices. Now suppose that $H$ contains a $P_5$. If $H$ has another connected component on more than one vertex then \contractd\ is $\mathsf{NP}$-hard by Theorem~\ref{thm:P6P5P2}. Otherwise, every other connected component of $H$ (if any) contains exactly one vertex in which case \contractd\ is polynomial-time solvable on $H$-free graphs by Theorem~\ref{theorem:p5free} and Proposition~\ref{prop:boundedtdom}. Assume now that $H$ contains a $P_4$. Then if $H$ has another connected component isomorphic to $P_4$, \contractd\ is $\mathsf{coNP}$-hard by Theorem~\ref{thm:2P4}. Otherwise, every other connected component has at most three vertices and we conclude by Theorem~\ref{thm:P4kP3} that \contractd\ is polynomial-time solvable on $H$-free graphs. Finally, if the longest path in $H$ has length at most two then we also conclude by Theorem~\ref{thm:P4kP3} that \contractd\ is polynomial-solvable on $H$-free graphs, which concludes the proof.


\section{Conclusion}
\label{sec-conclusion}

In this paper, we considered the problem of deciding whether the total domination number of a given graph $G$ can be reduced using exactly one edge contraction (called \contractd\ ). We recall that if we were allowed to use 3 edge contractions, the answer to this problem would always be yes due to a result of Huang et al. (\cite{HX10}, Theorem 4.3). We focused on several graph classes and determined the computational complexity of this problem. By putting together these results, we managed to obtain a complete dichotomy for $H$-free graphs (see Theorem 3).
 
In \cite{PAL-GAL-RIE}, the same problem was considered with respect to the domination number (called \textsc{1-Edge Contraction($\gamma$)}). Here the authors provided an almost dichotomy for $H$-free graphs (see Theorem 2). As mentioned in the introduction, the remaining cases left open have recently been solved. Interestingly, the two problems do not behave the same way on $H$-free graphs from a complexity point of view. This is even more interesting since it has been shown in  \cite{semidom} that the complexities of \textsc{Dominating set} and \textsc{Total dominating set} agree on  $H$-free graphs for any graph $H$.
 
In fact, it was even shown in \cite{semidom} that the complexities of \textsc{Dominating set}, \textsc{Semi-total dominating set} (given a graph $G$ and an integer $k$, does there exist a dominating set $S\subseteq V(G)$ with $\vert S\vert\leq k$ such that every vertex in $S$ is at distance at most two to another vertex in $S$) and \textsc{Total dominating set} agree on  $H$-free graphs for any graph $H$. Thus, it would be of interest to look at our problem with respect to the semi-total domination number and find out how it behaves on $H$-free graphs. A few of our results can be adapted to the semi-total case, but in order to obtain a complete dichotomy for $H$-free graphs, new approaches are needed.


\bibliographystyle{siam}
\bibliography{bibliography}

\begin{thebibliography}{10}

\bibitem{BTT11}
{\sc C.~Bazgan, S.~Toubaline, and Z.~Tuza}, {\em The most vital nodes with
  respect to independent set and vertex cover}, Discrete Applied Mathematics,
  159 (2011), pp.~1933--1946.

\bibitem{bazgan2013critical}
{\sc C.~Bazgan, S.~Toubaline, and D.~Vanderpooten}, {\em Critical edges for the
  assignment problem: Complexity and exact resolution}, Operations Research
  Letters, 41 (2013), pp.~685--689.

\bibitem{RBPDCZ10}
{\sc C.~Bentz, C.~Marie-Christine, D.~de~Werra, C.~Picouleau, and B.~Ries},
  {\em Blockers and transversals in some subclasses of bipartite graphs: when
  caterpillars are dancing on a grid}, Discrete Mathematics, 310 (2010),
  pp.~132 -- 146.

\bibitem{Bentz}
\leavevmode\vrule height 2pt depth -1.6pt width 23pt, {\em Weighted
  Transversals and Blockers for Some Optimization Problems in Graphs}, Progress
  in Combinatorial Optimization, ISTE-WILEY, 2012, pp.~203--222.

\bibitem{CWP11}
{\sc M.-C. Costa, D.~de~Werra, and C.~Picouleau}, {\em Minimum d-blockers and
  d-transversals in graphs}, Journal of Combinatorial Optimization, 22 (2011),
  pp.~857--872.

\bibitem{Di05}
{\sc R.~Diestel}, {\em Graph Theory}, vol.~173 of Graduate Texts in
  Mathematics, Springer, Heidelberg; New York, fourth~ed., 2010.

\bibitem{DPPR15}
{\sc {\"O}.~Y. Diner, D.~Paulusma, C.~Picouleau, and B.~Ries}, {\em Contraction
  blockers for graphs with forbidden induced paths}, in Algorithms and
  Complexity, Springer International Publishing, 2015, pp.~194--207.

\bibitem{diner2018contraction}
{\sc {\"O}.~Y. Diner, D.~Paulusma, C.~Picouleau, and B.~Ries}, {\em Contraction
  and deletion blockers for perfect graphs and {H}-free graphs}, Theoretical
  Computer Science, 746 (2018), pp.~49 -- 72.

\bibitem{PAL-GAL-RIE}
{\sc E.~Galby, P.~T. Lima, and B.~Ries}, {\em Reducing the domination number of
  graphs via edge contractions and vertex deletions}, to appear in Discrete
  Mathematics,  (2020).

\bibitem{TODO}
{\sc E.~Galby, F.~Mann, and B.~Ries}, {\em Blocking the domination number in
  $h$-free graphs}, manuscript,  (2020).

\bibitem{semidom}
{\sc E.~Galby, A.~Munaro, and B.~Ries}, {\em Semitotal domination: New hardness
  results and a polynomial-time algorithm for graphs of bounded mim-width},
  Theoretical Computer Science, 814 (2020).

\bibitem{garey}
{\sc M.~R. Garey and D.~S. Johnson}, {\em Computers and Intractability; A Guide
  to the Theory of NP-Completeness}, W. H. Freeman \& Co., New York, NY, USA,
  1990.

\bibitem{HX10}
{\sc J.~Huang and J.-M. Xu}, {\em Domination and total domination contraction
  numbers of graphs}, Ars Combinatoria, 94 (2010), pp.~431--443.

\bibitem{PBP}
{\sc F.~Mahdavi~Pajouh, V.~Boginski, and E.~Pasiliao}, {\em Minimum vertex
  blocker clique problem}, Networks, 64 (2014), pp.~48--64.

\bibitem{1IN3}
{\sc C.~Moore and J.~M. Robson}, {\em Hard tiling problems with simple tiles},
  Discrete \& Computational Geometry, 26 (2001), p.~573–590.

\bibitem{PPR16}
{\sc D.~Paulusma, C.~Picouleau, and B.~Ries}, {\em Reducing the clique and
  chromatic number via edge contractions and vertex deletions}, in ISCO 2016,
  vol.~9849 of LNCS, 2016, pp.~38--49.

\bibitem{paulusma2017blocking}
\leavevmode\vrule height 2pt depth -1.6pt width 23pt, {\em Blocking independent
  sets for {H}-free graphs via edge contractions and vertex deletions}, in TAMC
  2017, vol.~10185 of LNCS, 2017, pp.~470--483.

\bibitem{paulusma2018critical}
{\sc D.~Paulusma, C.~Picouleau, and B.~Ries}, {\em Critical vertices and edges
  in {H}-free graphs}, Discrete Applied Mathematics, 257 (2019), pp.~361 --
  367.

\end{thebibliography}

\end{document}